\documentclass[12pt]{article}
\usepackage{epsfig,amsmath,amssymb,amsfonts,amstext,amsthm,mathrsfs}
\usepackage{latexsym,graphics,epsf,epsfig,psfrag}
\usepackage{cite}
\usepackage{epstopdf,xcolor,ulem}
\newcommand{\goto}{\rightarrow}

\newtheorem{theorem}{\textbf{Theorem}}
\newtheorem{proposition}{\textbf{Proposition}}

\newcommand{\nn}{\nonumber}

\newcommand{\mE}{\mathbb{E}}
\newcommand{\mW}{\mathcal{M}}

\newcommand{\cC}{\mathcal{C}}

\DeclareMathAlphabet{\matheuf}{U}{euf}{m}{n}

\newcommand\redsout{\bgroup\markoverwith{\textcolor{red}{\rule[0.5ex]{2pt}{0.8pt}}}\ULon}
\newcommand{\zsfd}[1]{\ifmmode\text{\redsout{\ensuremath{#1}}}\else\redsout{#1}\fi}

\begin{document}

\vspace*{-2cm}

\begin{center}
  \baselineskip 1.3ex {\Large \bf Degraded Broadcast Channel with Secrecy Outside a Bounded Range
  \let\thefootnote\relax\footnotetext{The material in this paper was presented
in part  at the IEEE Information Theory Workshop (ITW), Jerusalem, Israel, April 2015~\cite{Zou2015itw},  at the IEEE International Symposium on Information Theory (ISIT), Hong Kong, China, June
2015~\cite{Zou2015isit}, and at the IEEE Information Theory Workshop (ITW),  Cambridge, UK, September, 2016 \cite{Zou2016itw}.} \let\thefootnote\relax\footnotetext{The work of S. Zou and Y. Liang was supported by a National Science Foundation
CAREER Award under Grant CCF-10-26565 and by the National Science Foundation under Grant CCF-16-18127 and CNS-11-16932. The work of L. Lai was supported by a National Science Foundation CAREER Award under Grant CCF-13-18980 and the National Science Foundation under Grant CNS-13-21223. The work of H. V. Poor was supported by the National Science
Foundation under Grant CMMI-1435778. The work of S.  Shamai (Shitz) was supported by the
European Commission in the framework of the Network of Excellence in
Wireless COMmunications NEWCOM\#, and as of June 2016 by
the European Union's Horizon 2020 Research And Innovation Programme,
grant agreement no. 694630.}
}
\\
 \vspace{0.15in} Shaofeng Zou, Yingbin Liang, Lifeng Lai, H. Vincent Poor and Shlomo Shamai (Shitz)
\let\thefootnote\relax\footnotetext{Shaofeng Zou is with the Coordinated Science Laboratory,
University of Illinois at Urbana-Champaign, Urbana, IL 61801 USA (email: szou3@illinois.edu).
Yingbin Liang is with the Department of Electrical and Computer Engineering, The Ohio State University, Columbus, OH 43220 USA (email: liang.889@osu.edu).
Lifeng Lai is with the Electrical and Computer Engineering, University of California, Davis, CA 95616 USA (email: lflai@ucdavis.edu).
H. Vincent Poor is with the Department of Electrical Engineering, Princeton University, Princeton, NJ 08544 USA (email: poor@princeton.edu).
Shlomo Shamai (Shitz) is
with the Department of Electrical Engineering, Technion-Israel
Institute of Technology, Technion City, Haifa 32000 Israel (email: sshlomo@ee.technion.ac.il).}
\end{center}

\begin{abstract}
The $K$-receiver degraded broadcast channel with secrecy outside a bounded range is studied, in which a transmitter sends $K$ messages to $K$ receivers, and the channel quality gradually degrades from receiver $K$ to receiver 1. Each receiver $k$ is required to decode message $W_1,\ldots,W_k$,  for $1\leq k\leq K$, and to be kept ignorant of $W_{k+2},\ldots,W_K$, for $k=1,\ldots, K-2$. Thus, each message $W_k$ is kept secure from receivers with at least two-level worse channel quality, i.e., receivers 1, $\ldots$, $k-2$. The secrecy capacity region is fully characterized. The achievable scheme designates one superposition layer to each message with binning employed for each layer. Joint embedded coding and binning are employed to protect all upper-layer messages from lower-layer receivers. Furthermore, the scheme allows adjacent layers to share rates so that part of the rate of each message can be shared with its immediate upper-layer message to enlarge the rate region. More importantly, an induction approach is developed to perform Fourier-Motzkin elimination of $2K$ variables from  the order of $K^2$ bounds to obtain a close-form achievable rate region. An outer bound is developed that matches the achievable rate region, whose proof involves recursive construction of the rate bounds and exploits the intuition gained from the achievable scheme.
\end{abstract}

%\vspace{-0.2in} \baselineskip\normalbaselineskip
\baselineskip 3.ex
\section{Introduction}
The broadcast channel models an important type of scenarios in which the transmitter's signal can simultaneously reach multiple receivers, and it has been widely used in wireless communications. Within the communication range of the transmitter, some receivers are intended while some are non-intended or even eavesdroppers from which the messages should be kept secure.
Due to this broadcast nature of wireless communications, security has arisen as an important issue. Various broadcast channel models with different transmission reliability constraints (i.e., legitimate receivers should decode messages destined for them) and different secrecy constraints (i.e., eavesdroppers should be kept ignorant of messages) have been intensively studied (see recent surveys \cite{Liang2009,Bloch2011,zou2015broadcast,baldi2016physical,poor2017wireless,hyadi2016overview}).

The basic broadcast channel with the secrecy constraint was the wiretap channel initiated by Wyner in \cite{Wyner1975}, in which a transmitter has a message intended for a legitimate receiver and wishes to keep this message secure from an eavesdropper. Csisz\'ar and K\"orner further generalized this model to the case with one more common message intended for both the legitimate receiver and the eavesdropper in \cite{Csiszar1978}.

These broadcast models were further  generalized to the multi-receiver case in \cite{Ekrem2009_s} and \cite{Liu2010}, in which a transmitter has a number of messages intended for a set of receivers, and all messages need to be secure from an eavesdropper. Another class of extension is the broadcast channel with layered decoding and layered secrecy \cite{Liu2010,ekrem2012,Zou2014it}, in which the transmitter has a number of messages intended for a set of receivers and as the channel quality of a receiver gets one level better, one more message is required to be decoded, and this message is required to be secure from all receivers with worse channel quality. {More specifically, a $K$-receiver broadcast channel is considered in \cite{Zou2014it} ($K=3$ in \cite{Liu2010} and \cite{ekrem2012}), in which a transmitter sends $K$ messages to $K$ receivers. The channel quality gradually degrades from receiver $K$ to receiver 1. Receiver $k$ is required to decode the first $k$ messages $W_1,\ldots,W_k$ for $1\leq k\leq K$, and to be kept ignorant of messages $W_{k+1},\ldots,W_K$  for $1\leq k\leq K-1$.}

%Further generalizations of these models with various decoding and secrecy constraints were studied in \cite{Ekrem2009_s,Liu2010,bagh2009,Benam2014,Khisti08juneit,Chia2012,Schaefer2014,Liuruoheng2009,Liuruoheng2010,ekrem2012,Zou2014it}. Among these studies, the broadcast channel with layered decoding and layered secrecy requirements was studied in \cite{Liu2010,ekrem2012,Zou2014it}, in which as the channel quality of a receiver gets one level better, one more message is required to be decoded, and this message is required to be secure from all receivers with worse channel quality.

{We note that  for the model considered in \cite{Liu2010,ekrem2012,Zou2014it}, the additional message decoded by a better receiver needs to be kept secure from the receiver with only one level worse channel quality. Here, any message $W_k$ should be decoded by receiver $k$, and be kept secure from receiver $k-1$. Such a model is well defined when receivers $k$ and $k-1$ have nonzero difference in channel quality so that nonzero secrecy rate can be achieved for $W_k$. However, such a model is not useful if the difference in channel quality between the adjacent receivers becomes asymptotically small (i.e., close to zero), because essentially no secrecy rate can be achieved under the secrecy requirement of the model. For example, consider a fading broadcast channel, in which the channel to each receiver is determined by a channel gain coefficient with amplitude $h$, where $h$ is continuous (the larger $h$, the better the channel). Here, the channel gains between two adjacent receivers can be arbitrarily close, and hence zero secrecy rate can be achieved for a message required to be decoded by one receiver and secured from the other receiver.}

{In this paper, we are interested in a model in which any message decoded at a certain receiver is not required to be kept secure from the one-level-worse receiver, but kept secure from the {\em multiple (possibly infinite)}-level-worse receiver. Such a model is valid as long as the ``the multiple (possibly infinite) levels" create nonzero differences in channel quality between receivers. In the fading channel, such a model captures scenarios in which messages intended for receivers with $h>h_0$ be kept secure from receivers with $h<h_0-\Delta$, i.e., the messages are not necessarily kept secure from receivers with channel quality between $h_0-\Delta$ and $h_0$. Here, $\Delta>0$ guarantees nonzero difference between receivers required to decode the messages and receivers required to be ignorant of the messages, so that nonzero secrecy rate can be achieved. We refer to such a secrecy requirement as \textit{secrecy outside a bounded range}.} 

{We note that although this paper focuses on the case with $\Delta$ corresponding to two levels of channel quality (as we describe below in more detail), the technical treatment here already contains all the necessary ingredients to design capacity-achieving secrecy schemes for the general case with secrecy outside arbitrary $m$ levels of channel quality. We discuss this generalization in Section \ref{sec:extension}. We also note that we recently applied/generalized this study to the fading channel in \cite{zou2017cns}. }

More formally, we consider the $K$-receiver degraded broadcast channel with secrecy outside a bounded range (see Fig.~\ref{fig:model}), in which a transmitter sends $K$ messages to $K$ receivers. The channel satisfies the degradedness condition, i.e., the channel quality gradually degrades from receiver $K$ to receiver 1. Furthermore, receiver $k$ is required to decode the first $k$ messages, $W_1,\ldots,W_k$, for $1\leq k\leq K$, and to be kept secure of $W_{k+2},\ldots,W_K$ for $k=1,\ldots, K-2$. Each message $W_k$ is required to be secure from the receiver $k-2$, which has two level worse channel quality, for $3\leq k\leq K$. In this way, the secrecy is required outside a range of two level channel quality. 

The main result of this paper lies in the complete characterization of the secrecy capacity region for the $K$-receiver degraded broadcast channel with secrecy outside a bounded range. To understand the challenges of the problem and the novelty of the paper, we first describe special cases, namely three-receiver and four-receiver models, studied by the authors in earlier conference versions \cite{Zou2015itw,Zou2015isit}. For three-receiver model, we show in \cite{Zou2015itw} that superposition of messages and joint binning and embedded coding (using lower layer messages to protect higher layer messages) achieves the secrecy capacity. However, in \cite{Zou2015isit} we show that a natural generalization of such a scheme does not provide the capacity region for the four-receiver model. A novel rate splitting and sharing scheme was proposed in \cite{Zou2015isit}, which is shown to be critical to further enlarge the achievable region and establish the secrecy capacity region for the four-receiver model. The idea is to first use lower-layer messages as random sources to protect higher-layer messages. {If the message at a certain layer (say layer $k$) is more than enough to protect the higher-layer messages, then such a message can also partially protect the message at layer $k$. Consequently, the protected message at layer $k$ can be shared between layer $k$ and its upper layer to enlarge the secrecy rate region.}
%and if the lower-layer messages are more than enough to protect the higher-layer messages, i.e., the lower-layer messages can protect not only the higher-layer messages but  also partial of their own rates, and then the lower-layer messages share the part of protected rates of their own to the higher-layer messages to enlarge the secrecy rate region.

%then further share remaining rate of lower-layer messages with upper-layer messages, as such part of lower layer messages can satisfy the same secrecy constraints as high-layer messages.

Further generalization of the capacity characterization for the above four-receiver model to the arbitrary $K$-user case becomes very challenging due to the following reasons. (1) Based on the understanding in the four-receiver model, each message as well as the random bin number at each layer can potentially serve as random source to protect all higher-layer messages (from lower layer receivers). The design of joint embedded coding and binning is very complicated to handle. For example, consideration of whether to adopt binning at layer $k$ depends on whether embedded coding of layer $k-1$ is sufficient to protect $W_k$ from receiver $k-2$, and whether embedded coding of layer $k-2$ and (possible) binning in layer $k-1$ are sufficient to protect $W_{k-1}$ and $W_{k}$ from receiver $k-3$, and so on. Incorporating all these considerations into the design of an achievable scheme is not feasible for arbitrary $K$-user model. (2) Due to rate splitting and sharing across adjacent layers, the rate region is expressed in terms of individual rate components. A typical technique to convert the rate region in terms of the (total) rate for each message is Fourier-Motzkin elimination. However, for the arbitrary $K$-user model, a large number of rate variables (more specifically, $2K$) should be eliminated from the order of $K^2$ rate bounds. Such procedure is not  analytically tractable in general. (3) Due to the reason that we employ joint embedded coding and binning to secure multiple messages, the analysis of secrecy guarantee is much more involved than the cases with only one or two messages secured by binning.

Despite the challenges mentioned above, in this paper, we fully characterize the secrecy capacity region for the $K$-receiver model with secrecy outside a bounded range. Our solution of the problem includes the following new ingredients. (1) Our achievable scheme employs binning in each layer, which avoids the complex consideration of whether or not it is necessary to employ binning for each layer. We also make an important observation that rate sharing only between adjacent layers is sufficient. This observation is critical to keep the obtained rate region simple enough for further manipulation. (2) We design an induction algorithm to perform Fourier-Motzkin elimination. Instead of directly eliminating $2K$ variables from the order of $K^2$ rate bounds, we eliminate a pair of variables at a time. We then further show that the region after each elimination step possesses a common structure by induction. (3) {In order to obtain the strong secrecy guarantee for the case with arbitrary $K$ users, we generalize the arguments in  \cite{han1993approximation,bloch2013strong,hou2013informational,hou2014effective,hayashi2006general} in which strong secrecy is obtained through channel resolvability.}
(4) Our development of the converse proof involves recursive construction of upper bound on the rate of each message such that proper terms cancel out across adjacent messages, and manipulation of the terms by exploiting intuitions in achievable schemes.

The remainder of this paper is organized as follows. In Section \ref{sec:model}, we introduce our system model. In Section \ref{sec:motivating}, we present two example models with three receivers and four receivers, respectively, which motivate the design of the achievable scheme for the model with arbitrary $K$ receivers. In Section \ref{sec:extension}, we discuss potential extensions of our results.
In Section \ref{sec:mainresults}, we present our main results for the model with arbitrary $K$ receivers.  Finally, in Section \ref{sec:con}, we conclude our paper.

\section{Channel Model}\label{sec:model}
\begin{figure}[htb]
  \centering
  \includegraphics[width=10cm]{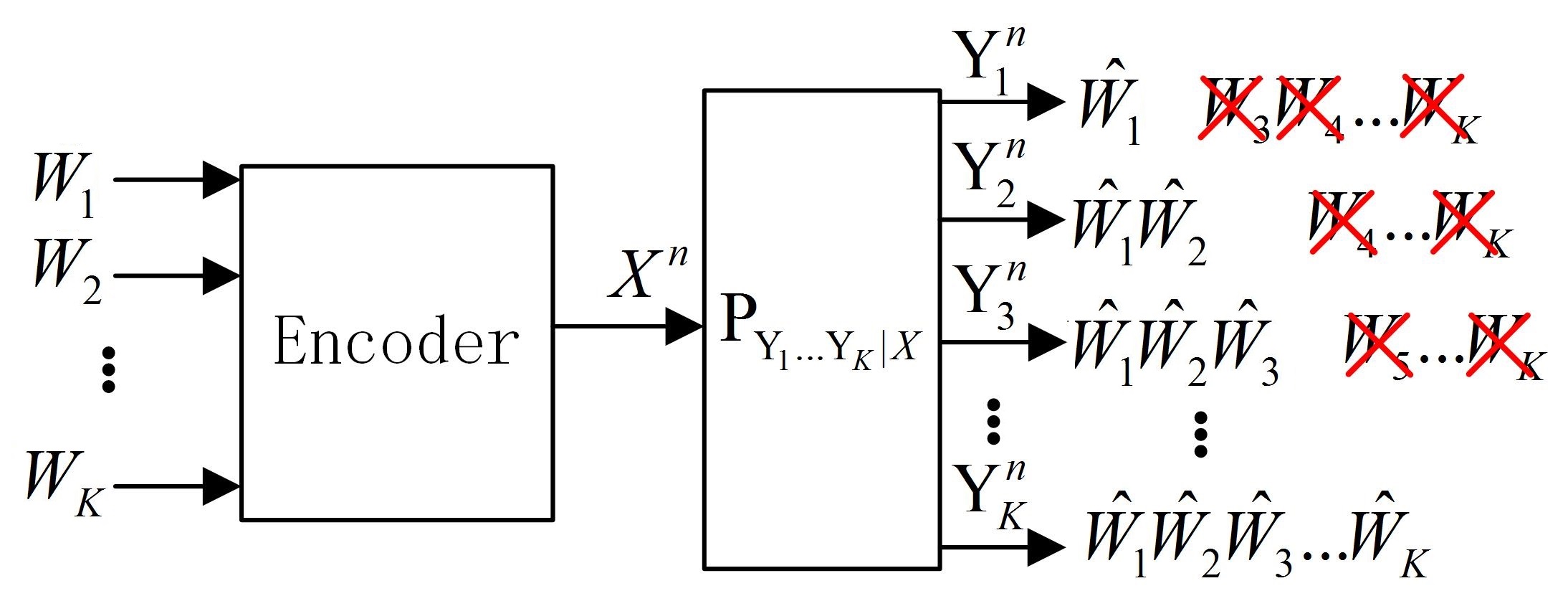}\\
  \caption{The $K$-receiver broadcast channel with secrecy outside a bounded range}\label{fig:model}
\end{figure}

In this paper, we consider a $K$-receiver degraded broadcast channel model with secrecy outside a bounded range (see Fig.~\ref{fig:model}). A transmitter sends information to $K$ receivers through a discrete memoryless channel. The channel transition probability function is $P_{Y_1\cdots Y_K|X}$, where $X\in \mathcal X$ denotes the channel input, and $Y_k\in\mathcal Y_k$ denotes the channel output at receiver $k$, for $1\leq k\leq K$. The channel is assumed to be degraded, i.e., the following Markov chain condition holds:
\begin{flalign}\label{eq:Markov}
  X\rightarrow Y_K\rightarrow Y_{K-1}\rightarrow \cdots\rightarrow Y_1.
\end{flalign}
Hence, the channel quality gradually degrades from receiver $K$ to receiver 1.
There are in total $K$ messages $W_1,W_2,\ldots,W_K$ intended for $K$ receivers with the following decoding and secrecy requirements. Receiver $k$ is required to decode messages $W_1,W_2,\ldots,W_k$, for $k=1,2,\ldots,K$, and to be kept secure of $W_{k+2},\ldots,W_K$, for $k=1,\ldots,K-2$ (see Fig.~\ref{fig:model}).

A $(2^{nR_1},\cdots,2^{nR_K},n)$ code for the channel consists of
\begin{list}{$\bullet$}{\topsep=0ex \leftmargin=6.5mm \rightmargin=0mm \itemsep=0mm}
\item $K$ message sets: $W_k\in\mathcal{W}_k=\{1,\cdots, 2^{nR_k}\}$ for $k=1,\cdots, K$, which are independent from each other and each message is uniformly distributed over the corresponding message set;
\item A (possibly stochastic) encoder $f^n$: $\mathcal{W}_1\times\cdots\times\mathcal{W}_K\rightarrow \mathcal{X}^n$ that maps a message tuple to an input $x^n$;
\item $K$ decoders $g_k^n: \mathcal{Y}_k^n\rightarrow (\mathcal{W}_1,\cdots,\mathcal{W}_k)$  that maps an output $y_k^n$ to a message tuple $(\hat w_1,\ldots,\hat w_k)$ for $k=1,\cdots, K$.
\end{list}

{A rate tuple $(R_1,\cdots,R_K)$ is said to be {\em achievable}, if there exists a sequence of $(2^{nR_1},\cdots,2^{nR_K},n)$ codes such that the average error probability
\begin{equation}
P_e^n=\text{Pr}\left( \cup_{k=1}^K \{(W_1,\cdots,W_k)\neq g_k^n(Y_k^n)\}\right) \rightarrow 0, \text{ as }n\rightarrow\infty,\label{eq:decode}
\end{equation}
and the secrecy metric at receiver $k$
\begin{equation}
I(W_{k},\cdots,W_{K};Y_{k-2}^n)\rightarrow 0, \text{ as }n\rightarrow \infty,\label{eq:security}
\end{equation}
for $k=3,\ldots,K$.
Here, we consider the strong secrecy metric instead of the weak secrecy metric as in \cite{Wyner1975,Csiszar1978}, which requires the mutual information in \eqref{eq:security} averaged over the block length $n$ go to zero as $n$ goes to infinity. The results in this paper may also be extended to an even stronger security notion as the semantic security \cite{thangaraj2014coding}, which enables quantifying the security of codes at finite lengths and is of practical importance in cryptography.}

The asymptotically small error probability in \eqref{eq:decode} implies that each receiver $k$ is able to decode messages $W_1,\ldots,W_k$, and asymptotically small secrecy metric in \eqref{eq:security} for each receiver $k$ implies that $W_k,\ldots,W_K$ and $Y_{k-2}^n$ are asymptotically independent, i.e., receiver $k$ is kept ignorant of messages $W_{k+2},\ldots,W_K$. Our goal is to characterize the \emph{secrecy capacity region} that consists of all achievable rate tuples.

\section{Motivating Examples}\label{sec:motivating}
In this section, we study two motivating examples with $K=3$ and $K=4$. The purpose is to motivate the development of the optimal achievable scheme for the case with arbitrary $K$ receivers step by step. More specifically, we study the example with three receivers to introduce the technique of joint embedded coding and binning. We study the example with four receivers to introduce the technique of rate splitting and sharing. These schemes turn out to be necessary to achieve the secrecy capacity region for the case with arbitrary $K$ receivers.
\subsection{Lessons Learned from $K=3$}\label{sec:31}
\begin{figure}[htb]
  \centering
  \includegraphics[width=8cm]{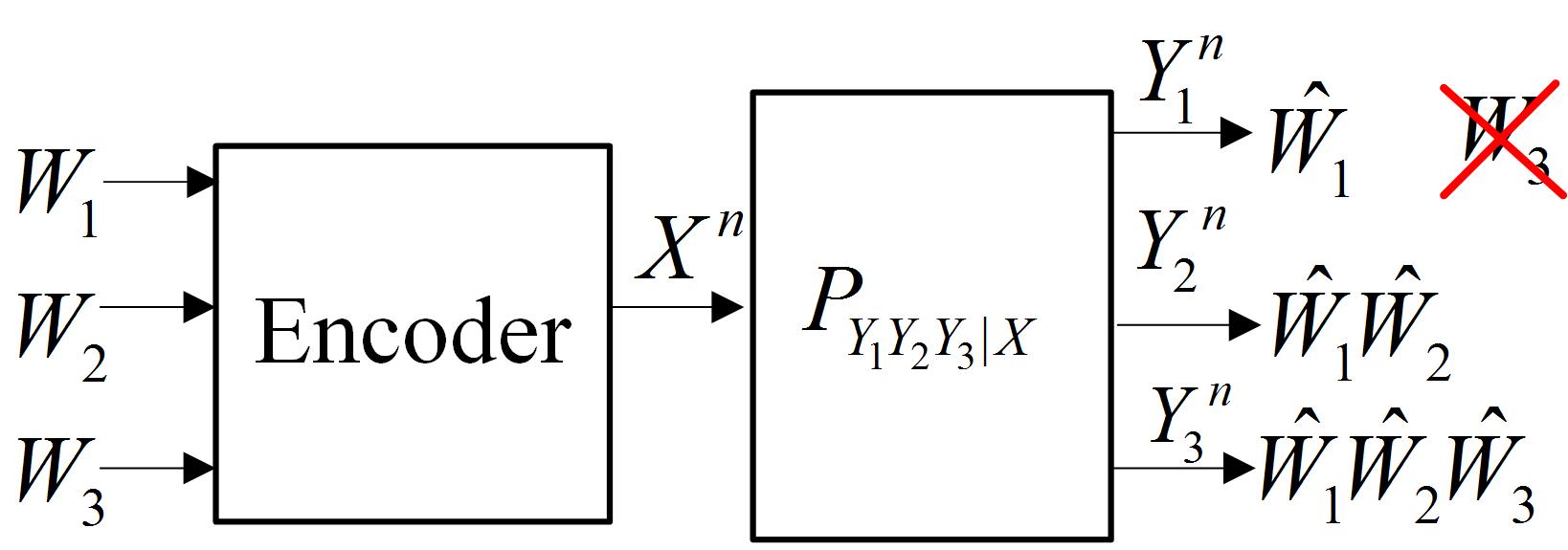}\\
  \caption{The three-receiver broadcast channel with secrecy outside a bounded range}\label{fig:model3}
\end{figure}
We start with the case in which there are three receivers (see Fig.~\ref{fig:model3}).
In this case, receiver 1 is required to decode $W_1$, receiver 2 is required to decode $W_1,W_2$, and receiver 3 is required to decode $W_1,W_2,W_3$. The system is also required to satisfy the secrecy constraint that the message $W_3$ is kept secure from receiver 1.

For such a model, a natural idea is to design superposition coding for encoding three messages $W_1,W_2,W_3$ into three layers, and then apply binning in the third layer to protect $W_3$ from receiver 1. However, such a scheme is suboptimal because it ignores an important fact that the random message $W_2$, which is not required to be decoded by receiver 1, can provide additional randomness to protect $W_3$ from receiver 1. This is referred to as \emph{embedded coding}.
In fact, if such a random source of $W_2$ is sufficient to protect $W_3$ from receiver 1, binning is not necessary.  If this is not sufficient to protect $W_3$, we apply binning in the third layer to further protect $W_3$ from receiver 1. The novelty of such an achievable scheme lies in exploiting the superposition layer of $W_2$ as embedded coding in addition to the binning scheme to protect $W_3$. Such a scheme turns out to achieve the secrecy capacity region as characterized in the following proposition.

\begin{proposition}\label{th:capadmc}
Consider the three-receiver degraded broadcast channel with secrecy outside a bounded range as described in Section \ref{sec:model}. The secrecy capacity region contains rate tuples $(R_1,R_2,R_3)$ satisfying
\begin{flalign}\label{eq:capacity_3r}
    R_1&\leq I(U_1;Y_1),\nn\\
    R_2&\leq I(U_2;Y_2|U_1),\nn\\
    R_3&\leq \min\{0,I(U_2;Y_2|U_1)-I(X;Y_1|U_1)\}+I(X;Y_3|U_2)
\end{flalign}
for some $P_{U_1U_2X}$ such that the following Markov chain condition holds
\begin{equation}
    U_1\rightarrow U_2\rightarrow X\rightarrow Y_3\rightarrow Y_2\rightarrow Y_1.
\end{equation}
\end{proposition}
\begin{proof}
  The proof can be found in \cite{Zou2015itw}.
  %is omitted here. We provide detailed proof for the  general case with arbitrary $K$ receivers in Theorem \ref{thm:capacity}.
\end{proof}

The idea of the achievable scheme is also reflected in the expression of the capacity region in \eqref{eq:capacity_3r}. The two bounds in ``min"  are corresponding to the two cases with the second layer of $W_2$ being sufficient and insufficient to protect $W_3$, respectively. If $I(U_2;Y_2|U_1)>I(X;Y_1|U_1)$, the randomness of $W_2$ is sufficient to exhaust receiver 1's decoding capability, and hence is good enough for protecting $W_3$. Thus, in this case, no binning is required in layer 3, and $R_3\leq I(X;Y_3|U_2)$.  On the other hand, if $I(U_2;Y_2|U_1)\leq I(X;Y_1|U_1)$, binning is required at layer 3 to protect $W_3$ in addition to randomness of $W_2$, and hence, $R_3\leq I(U_2;Y_2|U_1)-I(X;Y_1|U_1)+I(X;Y_3|U_2)$.

We note that a graphical representation of rate and equivocation quantities for the scalar Gaussian broadcast channel with secrecy outside a bounded range ($K=3$) is presented in \cite{bustin2015mmse}, which is based on the fundamental relationship between the mutual information and the minimum mean square error (MMSE) (I-MMSE approach \cite{guo2005mutual}).

%This can also be written as $R_3\leq I(X;Y_3|U_2)-I(X;Y_1|U_2)+I(U_2;Y_2|U_1)-I(U_2;Y_1|U_1)$, which has a clear intuitive interpretation. If receiver 1 know the message $W_1,W_2$ (i.e., $U_1,U_2$), the secrecy rate of $W_3$ will be $I(U_2;Y_2|U_1)-I(X;Y_1|U_1)$. But part of $U_2$ is secure from receiver 1 with rate $I(U_2;Y_2|U_1)-I(U_2;Y_1|U_1)$, which can be used to convey further secrecy rate for $W_3$.
\subsection{Lessons Learned from $K=4$}\label{sec:32}
\begin{figure}[htb]
  \centering
  \includegraphics[width=8cm]{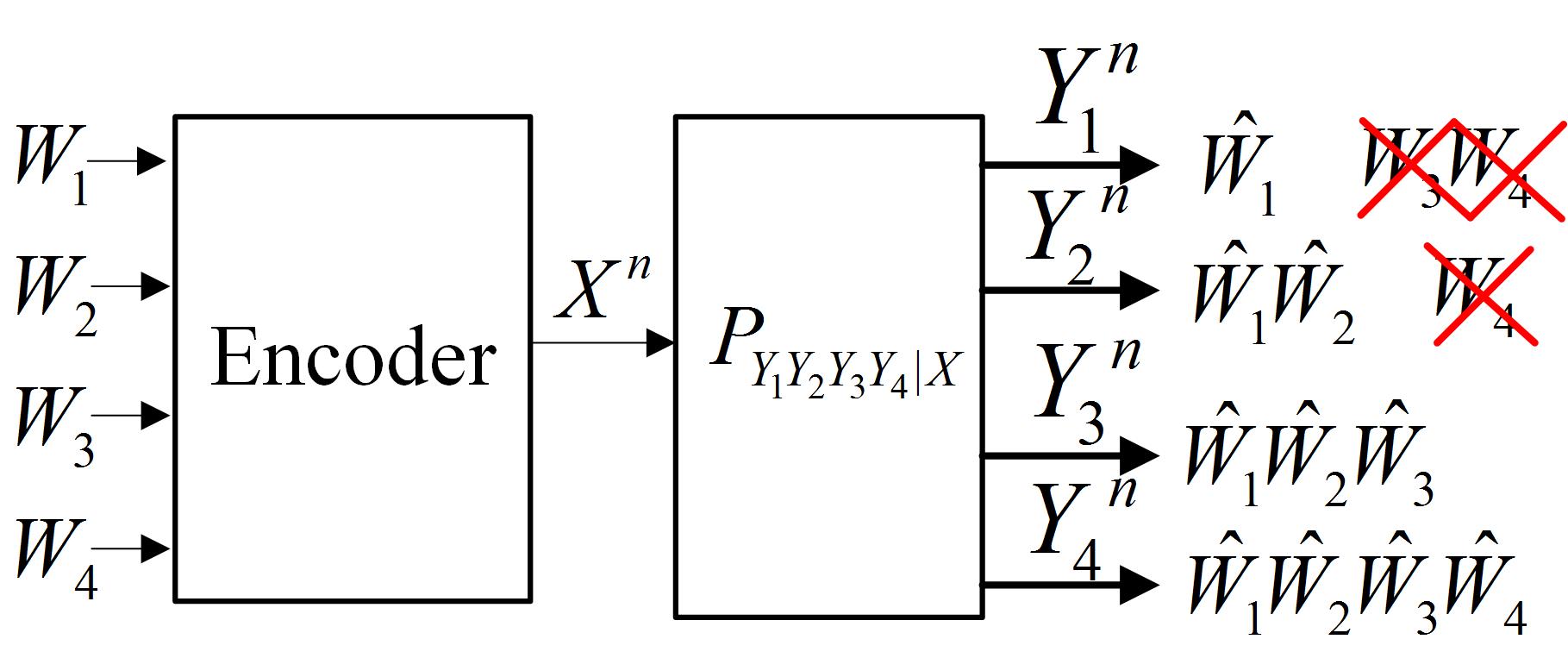}\\
  \caption{The four-receiver broadcast channel with secrecy outside a bounded range}\label{fig:model4}
\end{figure}
In this subsection, we study the model with four receivers (see Fig.~\ref{fig:model4}). In this model, receiver $k$ is required to decode messages $W_1,\ldots,W_k$, for $1\leq k\leq 4$. Furthermore, the message $W_3$ is required to be secure from receiver 1, and the message $W_4$ is required to be secure from receivers 1 and  2.

Although this four-receiver model seems to be a straightforward generalization of the three-receiver model, our exploration turns out to show that the achievable scheme for the three-receiver model is not sufficient to establish the secrecy capacity region for the four-receiver model. In order to understand this, we note that a direct generalization of the achievable scheme for the three-model involves first applying superposition coding to encode the four messages, and then  use the random message $W_3$ as embedded coding together with the binning in layer 4 (if necessary) to protect $W_4$, and use the random message $W_2$ as embedded coding together with the binning in layer 3 and layer 4 (if necessary) to protect $W_3$ and $W_4$.
Such a scheme then yields an achievable region with rate tuples $(R_1,R_2,R_3,R_4)$ satisfying
\begin{flalign}\label{achiv_itw}
    R_1&\leq I(U_1;Y_1),\nn\\
    R_2&\leq I(U_2;Y_2;U_1),\nn\\
    R_3&\leq I(U_3;Y_3|U_2) +\min\Big(0, I(U_2;Y_2|U_1)-I(U_3;Y_1|U_1)\Big),\nn \\
    R_4&\leq I(X;Y_4|U_3),\nn\\
    R_4&\leq I(X;Y_4|U_3)+I(U_3;Y_3|U_2)-I(X;Y_2|U_2),\nn\\
    R_3+R_4&\leq I(U_3;Y_3|U_2)+I(X;Y_4|U_3)+I(U_2;Y_2|U_1)-I(X;Y_1|U_1),
\end{flalign}
for some $P_{U_1U_2U_3X}$ satisfying the Markov chain condition $U_1\to U_2\to U_3\to X\to Y_4\to \cdots\to Y_1$. It turns out to be  very difficult to develop the converse proof for the bound $R_4\leq I(X;Y_4|U_3)$ in the above region. Thus, the optimality of the region \eqref{achiv_itw} cannot be guaranteed.

The major novelty in our scheme for this four-receiver model lies in the design of rate splitting and sharing, which helps enlarge the achievable region and thus establish the secrecy capacity region.
%In order to further enlarge the above region,  we incorporate the newly designed ingredient of rate splitting and sharing.
More specifically, if $W_3$ is sufficient to protect $W_4$, we further split $W_3$ into two parts, i.e., $W_{3,1}$ and $W_{3,2}$, such that $W_{3,1}$ serves as a random source to secure both $W_{3,2}$ and $W_4$ from receiver 2. Thus,$W_{3,2}$ satisfies both the decoding and secrecy requirements for $W_4$, and hence the rate of $W_{3,2}$ can be counted towards the rate of either $W_3$ or $W_4$. In this way, the achievable region can be enlarged. In fact, such an enlarged region is shown to be the secrecy capacity region as characterized in the following proposition.
\begin{proposition}\label{thm:capacity_4r}
Consider the four-receiver degraded broadcast channel with secrecy outside a bounded range as described in Section \ref{sec:model}. The secrecy capacity region consists of rate tuples $(R_1,R_2,R_3,R_4)$ satisfying
\begin{flalign}\label{capa}
R_1&\leq I(U_1;Y_1),\nn\\
R_2&\leq I(U_2;Y_2|U_1),\nn\\
R_3&\leq I(U_3;Y_3|U_2) +\min\Big(0, I(U_2;Y_2|U_1)-I(U_3;Y_1|U_1)\Big),\nn\\
R_4&\leq I(X;Y_4|U_3)+I(U_3;Y_3|U_2)-I(X;Y_2|U_2),\nn\\
R_3+R_4&\leq I(U_3;Y_3|U_2)+I(X;Y_4|U_3) +\min\Big(0,I(U_2;Y_2|U_1)-I(X;Y_1|U_1)\Big),
\end{flalign}
for some $P_{U_1U_2U_3X}$ such that the following Markov chain condition  holds
\begin{flalign}
U_1\rightarrow U_2\rightarrow U_3\rightarrow X\rightarrow Y_4\rightarrow Y_3\rightarrow Y_2\rightarrow Y_1.
\end{flalign}
\end{proposition}
\begin{proof}
  The proof can be found in \cite{Zou2015isit}.
  %is omitted here. We provide detailed proof for the  general case with arbitrary $K$ receivers in Theorem \ref{thm:capacity}.
\end{proof}

We note that  by using rate splitting and sharing, the bound  $R_4\leq I(X;Y_4|U_3)$ in the region \eqref{achiv_itw} is replaced by the bound $R_3+R_4\leq I(U_3;Y_3|U_2)+I(X;Y_4|U_3)$ in the region \eqref{capa}. Clearly, the region \eqref{capa} is larger than the region \eqref{achiv_itw} (for a given distribution of auxiliary random variables). Furthermore, the converse proof for the new bound on $R_3+R_4$ in  \eqref{capa} can be derived, and thus establishes the region \eqref{capa} as the secrecy capacity region.

Moreover, although we learn useful coding ingredients from the three-receiver and  four-receiver cases, direct generalization to arbitrary $K$-receiver model still gives rise to an analytically intractable achievable scheme. More specifically, the consideration of whether or not to use binning in the higher layers and whether or not to split and share the rates will be complex. For example, when $K=5$, whether to use binning in the fifth layer depends on whether the embedded coding in the third layer and (possibly) binning in the fourth layer are sufficient to protect $W_4,W_5$ from receiver 3 and whether the embedded coding in the fourth layer is sufficient to protect $W_5$ from receiver 3. Such considerations become intractable when $K$ is large. Thus, the major design issue for the arbitrary $K$-receiver case is to develop an achievable scheme that effectively incorporates the necessary coding ingredients as well as yielding a tractable rate region for analysis. This is the focus of the following section.

\section{Main Results}\label{sec:mainresults}
In this section, we first present our main result of characterization of the secrecy capacity region for the $K$-receiver model, and then describe the idea behind the design of the achievable scheme.
\subsection{Secrecy Capacity Region}
The following theorem states our main result. For simplicity of notation, we define $U_K=X$.
\begin{theorem}\label{thm:capacity}
  Consider the $K$-receiver degraded broadcast channel with secrecy outside a bounded range as described in Section \ref{sec:model}. The secrecy capacity region consists of rate tuples $(R_1,R_2,\ldots,R_K)$ satisfying
\begin{subequations}
\begin{flalign}
  R_1&\leq I(U_1;Y_1),\label{eq:capacity1}\\
  \sum_{j=2}^{k}R_j&\leq \sum_{j=2}^{k}I(U_j;Y_j|U_{j-1}),\quad\text{  for }2\leq k\leq K, \label{eq:capacity2}\\
 % \sum_{i=1}^{K}R_i &\leq \sum_{i=1}^{K-1}I(U_i;Y_i|U_{i-1})+I(X;Y_K|U_{K-1})\nn\\
  \sum_{j=l}^k R_j &\leq \left(\sum_{j=l-1}^{k}I(U_j;Y_j|U_{j-1})\right)-I(U_k;Y_{l-2}|U_{l-2}), \quad\text{ for } 3\leq l\leq k\leq K,\label{eq:capacity3}
\end{flalign}
\end{subequations}
for some $P_{U_1U_2\ldots U_{K}}$ such that the following Markov chain condition holds:
\begin{flalign}\label{eq:Markov1}
  U_1\goto U_2\goto \cdots\goto U_{K}\goto Y_K\goto\cdots\goto Y_2\goto Y_1.
\end{flalign}
\end{theorem}
\begin{proof}
  The proof of the achievability and the proof of converse are provided in Appendices \ref{app:achiv} and \ref{app:converse}, respectively.
\end{proof}

In the above capacity region, the bounds \eqref{eq:capacity1} and \eqref{eq:capacity2} are due to the decoding requirements, i.e., receiver $k$ should decode messages $W_1,\ldots,W_k$, for $1\leq k\leq K$. The sum rate bounds \eqref{eq:capacity2} are due to the rate sharing scheme we design. The bounds \eqref{eq:capacity3} are due to the secrecy requirements, i.e., messages $W_l,\ldots, W_k$ need to be kept secure from receiver $l-2$ for $3\leq l\leq k\leq K$. Furthermore, the bounds \eqref{eq:capacity3} can be further written as
\[\sum_{j=l}^k R_j \leq \sum_{j=l-1}^{k}\bigg(I(U_j;Y_j|U_{j-1})-I(U_j;Y_{l-2}|U_{j-1})\bigg),\]
which has  clear intuitive interpretation. The term $I(U_j;Y_j|U_{j-1})-I(U_j;Y_{l-2}|U_{j-1})$ is corresponding to the rate in layer $j$ that can be secure from receiver $l-2$ given the knowledge of layer $j-1$. Those rates $I(U_j;Y_j|U_{j-1})-I(U_j;Y_{l-2}|U_{j-1})$ for $l-1\leq j\leq k$ can all be counted towards $\sum_{j=l}^k R_j$ in accordance to the secrecy requirement of keeping $W_l,\ldots, W_k$ secured from receiver $l-2$.

If we set $K=3$ and $K=4$, the region in Theorem \ref{thm:capacity} reduces to equivalent but different forms from the regions in Proposition \ref{th:capadmc} and Proposition \ref{thm:capacity_4r}. The equivalence is justified by the converse proofs. However, the achievable schemes for the three-receiver model in Section \ref{sec:31} and the four-receiver model in Section \ref{sec:32} cannot be easily generalized to the arbitrary $K$-receiver model.

Our design of the achievable scheme for the general arbitrary $K$-receiver model is different from those for the three-receiver and four receiver models, and includes the following new ingredients. The scheme employs binning in each layer, which avoids the complex consideration of whether or not it is necessary to employ binning for each layer. The rate sharing scheme is limited only between adjacent layers which captures the essence of the problem and helps simplify the obtained rate region. Furthermore, we design an induction algorithm to perform Fourier-Motzkin elimination, which makes the problem of eliminating $2K$ variables from the order of $K^2$ bounds analytically tractable. These ideas are described in more detail in Subsection \ref{sec:sub}.

The converse for the achievable region can be developed. The bounds \eqref{eq:capacity1} and \eqref{eq:capacity2} can be derived following standard steps.
However, the proof for the bounds \eqref{eq:capacity3} is more involved and requires careful recursive construction of the terms such that proper terms cancel out across adjacent messages.

\subsection{Achievable Scheme}\label{sec:sub}

In this subsection, we introduce the idea of the achievable scheme for the arbitrary $K$-receiver model, which is based on superposition coding, binning, embedded coding, and rate splitting and sharing. We also sketch the novel induction idea to analyze Fourier-Motzkin elimination to characterize the achievable region.

\textbf{Superposition, binning, embedded coding.} We design one layer of codebook for each message, i.e., layer $k$ corresponds to $W_k$, for $1\leq k\leq K$.  To avoid the complex consideration of whether to use binning, we employ binning in each layer. We divide the codewords in each layer into a number of bins, where the bin number contains the information of the corresponding message. We use joint embedded coding and binning to provide randomness for secrecy.

\textbf{Rate splitting and sharing.} We design rate splitting and sharing to enlarge the achievable region.
More specifically, within the $k$-th layer, we split the message $W_k$ into two parts $W_{k,1}$ and $W_{k,2}$. The message $W_{k,1}$ serves as embedded coding which is a random source in addition to the binning to protect $W_{k,2}$ and the higher layer messages from receiver $Y_{k-1}$, i.e., we require that $(W_{k,2},W_{k+1,1},W_{k+1,2},\ldots,W_{K,1},W_{K,2})$ be secure from receiver $Y_{k-1}$, for $2\leq k\leq K-1$. Furthermore, the upstream receiver $Y_{k+1}$ can also decode $W_{k,2}$ because $Y_{k+1}$ has a better channel quality than $Y_k$. Thus, the message $W_{k,2}$ satisfies both the  decoding and secrecy requirements for message $W_{k+1}$, and hence, the rate of $W_{k,2}$ can be counted towards the rate of either $W_k$ or $W_{k+1}$. By such a rate sharing strategy, the achievable region is enlarged.

The rate can only be shared between adjacent receivers, which is an important observation of this problem, and is critical to reduce the complexity of the design of the rate splitting and sharing strategy. More specifically, the rate of $W_{k,2}$ cannot be counted towards the rates of $W_{k+2},\ldots,W_K$, because $W_{k+2},\ldots,W_K$ are required to be secure not only from receiver $Y_{k-1}$ but also from $Y_k$ that are required to decode $W_{k,2}$.

Based on the above achievable scheme, we obtain the following achievable region:
\begin{flalign}\label{eq:achivableregion}
  R_1&\leq I(U_1;Y_1),\nn\\
  R_{k,1}+R_{k,2}&\leq I(U_k;Y_k|U_{k-1})  \text{, for } 2\leq k\leq K,\nn\\
  R_{k-1,2}+\sum_{i=k}^j (R_{i,1}+R_{i,2})&\leq \sum_{i=k-1}^j I(U_i;Y_i|U_{i-1})-I(U_j;Y_{k-2}|U_{k-2}), \nn\\
   &\quad\text{ for } 3\leq k\leq K , k-1\leq j\leq K,
\end{flalign}
{where we use the convention that $\sum_{i=j}^k X_i=0$ if $j>k$.}
%\textbf{Rate sharing.}

The above region are expressed in terms of component rates due to rate splitting. In order to express the above region in terms of total rate for each message, we introduce the technique of rate sharing. We define $R_k=R_{k-1,2}+R_{k,1}$ for $3\leq k\leq K-1$, $R_2=R_{2,1}$ and $R_K=R_{K-1,2}+R_{K,1}+R_{K,2}$. We then wish to project the region \eqref{eq:achivableregion} onto the rate space $(R_1,\ldots,R_K)$. This can be done by adding the above rate definitions to the achievable region \eqref{eq:achivableregion} and then perform the Fourier-Motzkin elimination to eliminate $R_{k,1}$ and $R_{k,2}$ for $2\leq k\leq K$.

\textbf{Fourier-Motzkin elimination via induction.}
The total number of bounds in the achievable region \eqref{eq:achivableregion} is on the order of $K^2$ with $2K$ variables to be eliminated. Directly applying Fourier-Motzkin elimination is not analytically tractable. In order to solve this problem, we design the following induction algorithm to perform Fourier Motzkin elimination. We eliminate the rate pairs $R_{k-1,2}$ and $R_{k,1}$ for $3\leq k\leq K$ one at each step, and wish to show that the region $\mathcal R_k$ after eliminating $R_{k-1,2}$ and $R_{k,1}$ possesses the following structure:
\begin{flalign}
  R_1&\leq I(U_1;Y_1),\nn\\
  \sum_{i=2}^j R_i&\leq \sum_{i=2}^j I(U_i;Y_i|U_{i-1}), \text{ for } 2\leq j\leq k-1,\nn\\
  \sum_{i=2}^k R_i+R_{k,2}&\leq \sum_{i=2}^k I(U_i;Y_i|U_{i-1}),\nn\\
  \sum_{i=l}^j R_i&\leq \sum_{i=l-1}^j I(U_i;Y_i|U_{i-1})-I(U_j;Y_{l-2}|U_{l-2}),\nn\\
  &\hspace{3cm}\text{for } 3\leq l\leq j\leq k-1,\nn\\
  \sum_{i=l}^k R_i+R_{k,2}&\leq \sum_{i=l-1}^k I(U_i;Y_i|U_{i-1})-I(U_k;Y_{l-2}|U_{l-2}),\nn\\
  &\hspace{3cm}\text{ for } 3\leq l\leq k+1.
\end{flalign}
Such a claim can be easily verified for the case when  $k=3,4,5$. If such a claim holds for $\mathcal R_k$, we then are able to show (see Appendix \ref{app:Fourier} for detailed proof) that the region $\mathcal R_{k+1}$ after eliminating $R_{k,2}$ and $R_{k+1,1}$ possesses the same structure given by
\begin{flalign}
  R_1&\leq I(U_1;Y_1),\\
  \sum_{i=2}^j R_i&\leq \sum_{i=2}^j I(U_i;Y_i|U_{i-1}), \text{ for } 2\leq j\leq k,\\
  \sum_{i=2}^{k+1} R_i+R_{k+1,2}&\leq \sum_{i=2}^{k+1} I(U_i;Y_i|U_{i-1}),\\
  \sum_{i=l}^j R_i&\leq \sum_{i=l-1}^j I(U_i;Y_i|U_{i-1})-I(U_j;Y_{l-2}|U_{l-2}),\\
  &\hspace{3cm}\text{for } 3\leq l\leq j\leq k,\\
  \sum_{i=l}^{k+1} R_i+R_{k+1,2}&\leq \sum_{i=l-1}^{k+1} I(U_i;Y_i|U_{i-1})-I(U_{k+1};Y_{l-2}|U_{l-2}),\\
  &\hspace{3cm}\text{ for } 3\leq l\leq k+2.
\end{flalign}
The last step is to eliminate $(R_{K-1,2},R_{K,1},R_{K,2})$. Thus, the above induction algorithm and arguments yield the achievable region in Theorem~\ref{thm:capacity}.

\section{Extension to More General Cases}\label{sec:extension}

{In this paper, we have focused on the case with secrecy outside two levels of channel quality. In fact, such a case captures the essence of this type of model, and the design of the capacity-achieving secrecy schemes already consists of all the necessary ingredients to address the general case with secrecy outside arbitrary $m$ levels of channel quality, i.e., the techniques of joint embedded coding and binning, rate splitting and sharing, and inductive Fourier-Motzkin elimination.}

{For $m>2$ the rate splitting and sharing are more involved than for the case with $m=2$. Each message $W_k$ should be split into $m$ submessages, $W_{k,1},\ldots,W_{k,m}$. All the submessages in layers indexed from $k-m+1$ to $k-1$ but $\{W_{k-i,i+1}\}_{i=1}^{m-1}$ serve as embedded coding in addition to binning to protect $\{W_{k-i,i+1}\}_{i=1}^{m-1}$ and all higher-layer (with index no less than $k$) submessages from receiver $Y_{k-m}$. Here, we index the submessages such that $W_{k,i}$ is secured from receiver $Y_{k-m+i-1}$, for $ 1\leq i\leq m$ and $2\leq k\leq K$. The upstream receiver $Y_k$ can also decode all the submessages $\{W_{k-i,i+1}\}_{i=1}^{m-1}$. Hence, $\{W_{k-i,i+1}\}_{i=1}^{m-1}$ satisfy both the decoding and secrecy requirements for message $W_k$. Then the rates of  $\{W_{k-i,i+1}\}_{i=1}^{m-1}$ can be counted towards the rate of $W_k$. We then define the rate sharing such that $R_k=\sum_{i=0}^{m-1}R_{k-i,i+1}$.
Based on the above achievable scheme, we can obtain an achievable region in terms of $W_{k,j}$, for $1\leq k\leq K$ and  $1\leq j\leq m$. We then project this region onto the rate space $(R_1,\ldots,R_K)$. This can be done by a similar but rather complicated inductive Fourier-Motzkin elimination.}

The results here can be further generalized to models with continuously changing channel state parameters, e.g., Gaussian fading channel \cite{Liang2014_fading}, and Gaussian multiple input multiple output (MIMO) channel \cite{Shamai2003}. {For example, in our recent work \cite{zou2017cns}, we study the fading channel with secrecy outside a bounded range. More specifically, we first quantify the continuous channel state with infinitely many discrete channel states, and then apply/generalize the techniques in this paper.}

\section{Conclusion}\label{sec:con}

In this paper, we have studied the $K$-receiver degraded broadcast channel with secrecy outside a bounded range. We have proposed a novel achievable scheme based on superposition coding, joint embedded coding and binning, and rate splitting and sharing. The combination of embedded coding and binning to achieve secrecy captures the fact that lower-layer message can serve as embedded coding to protect higher-layer messages. Rate splitting and sharing are critical to enlarge the achievable region for which the converse proof can be established. Moreover, our design exploits an important property that the rate sharing should be only between adjacent receivers, which significantly reduces the complexity of the achievable scheme. We have further proposed a novel induction algorithm to perform Fourier-Motzkin elimination on the achievable region with $2K$ variables to be eliminated from the order of $K^2$ bounds. We have also constructed a converse proof, which involves careful recursive construction of rate bounds, and  exploits the intuition gained from embedded coding in the achievable scheme. By the converse proof, we have demonstrated the optimality of our achievable scheme and established the secrecy capacity region.

{This paper has focused on characterizing the information theoretic performance limits which is based on random coding arguments. It is of further interest to design practical coding schemes such as polar codes \cite{Koyluoglu2012,delOlmo2016arXiv,mahdavifar2011achieving,chou2016polar,gulcu2017achieving,wei2016polar,bloch2015error} and low density parity check (LDPC) codes \cite{klinc2011ldpc} to achieve secrecy capacity.}

\section*{{Acknowledgement}}
{The authors would like to thank the associate editor Dr. Matthieu Bloch for his constructive suggestions to obtain strong secrecy. They would also like to thank the anonymous reviewers for their insightful comments, which helped to significantly improve the presentation of the paper.}

\hspace{1cm}

\appendix

\noindent {\Large \textbf{Appendix}}
\section{Achievability Proof of Theorem \ref{thm:capacity}}\label{app:achiv}
The achievability proof is based on superposition coding, embedded coding, binning, rate splitting and sharing. We use random codes and fix a distribution $P_{U_1U_2\ldots,U_{K-1}X}P_{Y_1\ldots Y_K|X}$ satisfying the Markov chain condition in \eqref{eq:Markov1}.  Let $T_{\epsilon}^{n}(P_{U_1\ldots U_{K-1}XY_1\ldots Y_K})$ denote the strongly jointly $\epsilon$-typical set based on the fixed distribution \cite[Chapter 3]{massey1980applied}, \cite{orlitsky1995coding}. The achievable scheme is designed as follows:

\emph{Random codebook generation:}
For simplicity, we define $U_K=X$ in the following proof, i.e., $P_{U_1\cdots U_{K-1}X}=P_{U_1\cdots U_{K}}$.
\begin{list}{$\bullet$}{\topsep=0ex \leftmargin=5.5mm \rightmargin=0mm \itemsep=0mm}
    \item Generate $2^{nR_1}$ independent and identically distributed
    (i.i.d.) $u^n_1$ with distribution $\prod_{i=1}^{n}P(u_{1,i})$.
    Index these codewords as $u^n_1(w_1)$, $w_1 \in [1,2^{nR_1}]$.
    \item For each $u_1^n(w_1)$, generate $2^{n(R_{2,1}+R_{2,2})}$ i.i.d. $u_2^n$ by $\prod_{i=1}^n P(u_{2,i}|u_{1,i})$. Partition these codewords into $2^{nR_{2,2}}$ bins.
    Index these codewords as $u_2^n(w_1,w_{2,1},w_{2,2})$, $w_{2,1} \in [1,2^{nR_{2,1}}]$, $w_{2,2} \in [1,2^{nR_{2,2}}]$.
    \item For each $u_2^n(w_1,w_{2,1},w_{2,2})$, generate $2^{n(R_{3,1}+R_{3,2}+R_{3,3})}$ i.i.d. $u_3^n$ by $\prod_{i=1}^n P(u_{3,i}|u_{2,i})$. Partition these codewords into $2^{nR_{3,1}}$ bins, and further partition each bin into $2^{nR_{3,2}}$ sub-bins. Hence, there are $2^{nR_{3,3}}$  $u_3^n$ in each sub-bin.
        We use $w_{3,1}\in[1:2^{nR_{3,1}}]$ to denote the bin number, $w_{3,2}\in [1:2^{nR_{3,2}}]$ to denote the sub-bin number, and $l_3\in[1:2^{nR_{3,3}}]$ to denote the index within the bin. Hence, each $u_3^n$ is indexed by $(w_1,w_{2,1},w_{2,2},w_{3,1},w_{3,2},l_{3})$.
    \item For $4\leq k\leq K$, for each $u_{k-1}^n(w_1,\ldots, w_{k-1,1},w_{k-1,2},l_{k-1})$,  generate $2^{n(R_{k,1}+R_{k,2}+R_{k,3})}$ i.i.d. $u_k^n$ by $\prod_{i=1}^n P(u_{k,i}|u_{k-1,i})$. Partition these codewords into $2^{nR_{k,1}}$ bins, and further partition each bin into $2^{nR_{k,2}}$ sub-bins. Hence, there are $2^{nR_{k,3}}$  $u_k^n$ in each sub-bin.
        We use $w_{k,1}\in[1:2^{nR_{k,1}}]$ to denote the bin number, $w_{k,2}\in [1:2^{nR_{k,2}}]$ to denote the sub-bin number, and $l_k\in[1:2^{nR_{k,3}}]$ to denote the index within the bin. Hence, each $u_k^n$ is indexed by $(w_1,\ldots,w_{k-1,1},w_{k-1,2},l_{k-1},w_{k,1},w_{k,2},l_k)$.
\end{list}
The  codebook is revealed to both the transmitter and the receivers.

\emph{Encoding:}

To send a message tuple $(w_1,w_{2,1},w_{2,2},\ldots,w_{K,1},w_{K,2})$,
the transmitter randomly and uniformly generates $l_k\in[1:2^{nR_{k,3}}]$ for $3\leq k\leq K$, and sends $x^n(w_1,\ldots,w_{K,1},w_{K,2},l_3,\ldots,l_K)$.

{\em Decoding:}

\begin{list}{$\bullet$}{\topsep=0ex \leftmargin=5.5mm \rightmargin=0mm \itemsep=0mm}
  \item Receiver $1$ claims that $\widehat{w}_1$ is sent, if there exists a unique $\widehat{w}_1$ such that
  \[\Big(u_1^n(\widehat w_1), y_1^n\Big)\in T_{\epsilon}^{n}(P_{U_1Y_1}).\]
  Otherwise, it declares an error.
  \item Receiver $2$ claims that $(\widehat{w}_1,\widehat w_{2,1},\widehat w_{2,2})$ is sent, if there exists a unique tuple $(\widehat{w}_1,\widehat w_{2,1},\widehat w_{2,2})$ such that
  \[\Big(u_1^n(\widehat w_1),  u_2^n(\widehat{w}_1,\widehat{w}_{2,1},\widehat w_{2,2}),y_2^n\Big)\in T_{\epsilon}^{n}(P_{U_1U_2 Y_2}).\]
  Otherwise, it declares an error.
  \item For $3\leq k\leq K$, receiver $k$ claims that $(\widehat{w}_1,\ldots,\widehat w_{k,1},\widehat w_{k,2})$ is sent, if there exists a unique tuple $(\widehat{w}_1,\ldots,\widehat w_{k,1},\widehat w_{k,2},\widehat{l_3},\ldots,\widehat {l_k})$ such that
      \begin{flalign}
\Big(u_1^n(\widehat w_1),
\ldots,u_k^n(\widehat{w}_1,\ldots,\widehat w_{k,1},\widehat w_{k,2},\widehat{l_3},\ldots,\widehat {l_k}),y_k^n\Big)
\in T_{\epsilon}^{n}(P_{U_1\cdots U_k Y_k}).\nn
      \end{flalign}
      Otherwise, it declares an error.

\end{list}

{\em Analysis of error probability:} By the law of large numbers and the packing lemma \cite{Gamal2012}, receiver $k$ decodes the message $(w_1,\ldots,w_{k,1},w_{k,2})$ for $2\leq k\leq K$ and receiver 1 decodes the message $w_1$ with asymptotically small error probabilities if the following inequalities are satisfied:
\begin{flalign}\label{eq:decoding}
  R_1&\leq I(U_1;Y_1),\nn\\
  R_{2,1}+R_{2,2}&\leq I(U_2;Y_2|U_1),\nn\\
  R_{k,1}+R_{k,2}+R_{k,3}&\leq I(U_k;Y_k|U_{k-1}), \text{ for }3\leq k\leq K.
\end{flalign}

{{\em Analysis of secrecy:} We require that $W_{k-1,2}, W_{k,1},W_{k,2},\ldots,W_{K,1},W_{K,2}$ be secure from receiver $Y_{k-2}$ for $3\leq k\leq K$. It then suffices to show that
\begin{flalign}\label{eq:leakage}
  I\bigg(W_{k-1,2}&, W_{k,1},W_{k,2},\ldots,W_{K,1},W_{K,2};Y_{k-2}^n|\cC\bigg)\goto 0,\text{ as }n\goto\infty,
\end{flalign}
for $3\leq k\leq K$, where $\cC$ denotes a random codebook over the codebook ensemble.} This implies the existence of one codebook that guarantees secrecy.

{We note that $l_k$ in random codebook generation is a realization of the random variable $L_k$. For notational convenience, let $L_j^k=(L_j,\ldots,L_k)$, $l_j^k=(l_j,\ldots,l_k)$ for $3\leq j\leq k\leq K$ and $\mathcal M_k=(W_{k-1,2}, W_{k,1},W_{k,2},\ldots,W_{K,1},W_{K,2})$, for $3\leq k\leq K$. 
}
	
{By the independence of the messages, i.e., $\mW_k$ and $(W_1,\ldots,W_{k-2,1},W_{k-2,2},L_{k-2})$ are independent, and the fact that $U_{k-2}^n$ is a function of $(W_1,\ldots,W_{k-2,1},W_{k-2,2},L_{k-2})$ given $\cC$, it follows that $U_{k-2}^n$ is independent of $\mW_k$, and thus
\begin{flalign}
&I(\mW_k;Y_{k-2}^n|\cC)\nn\\
&=H(\mW_k|\cC)-H(\mW_k|Y_{k-2}^n,\cC)\nn\\
&=H(\mW_k|U_{k-2}^n,\cC) - H(\mW_k|Y_{k-2}^n,\cC)\nn\\
&\leq H(\mW_k|U_{k-2}^n,\cC) - H(\mW_k|Y_{k-2}^n,U_{k-2}^n,\cC)\nn\\
&\leq I(\mW_k;Y_{k-2}^n|U_{k-2}^n,\cC).
\end{flalign}
Connecting the idea of channel resolvability to secrecy \cite{hayashi2006general,bloch2013strong,hou2013informational,hou2014effective}, it follows that
\begin{flalign}\label{eq:ED}
&I(\mW_k;Y_{k-2}^n|U_{k-2}^n,\cC)\nn\\
&=\mE\log\frac{P(\mW_k,Y_{k-2}^n|U_{k-2}^n,\cC)}{P(\mW_k|U_{k-2}^n,\cC)P(Y_{k-2}^n|U_{k-2}^n,\cC)}\nn\\
&=\mE\log\frac{P(Y_{k-2}^n|\mW_k,U_{k-2}^n,\cC)}{P(Y_{k-2}^n|U_{k-2}^n,\cC)}\nn\\
&=\mE\left[\log\frac{P(Y_{k-2}^n|\mW_k,U_{k-2}^n,\cC)}{P(Y_{k-2}^n|U_{k-2}^n)} +\log\frac{P(Y_{k-2}^n|U_{k-2}^n)}{P(Y_{k-2}^n|U_{k-2}^n,\cC)}\right]\nn\\
&\leq \mE\left[\log\frac{P(Y_{k-2}^n|\mW_k,U_{k-2}^n,\cC)}{P(Y_{k-2}^n|U_{k-2}^n)} \right],
\end{flalign}
where the last step is due to the fact that
\begin{flalign}
\mE\left[\log\frac{P(Y_{k-2}^n|U_{k-2}^n)}{P(Y_{k-2}^n|U_{k-2}^n,\cC)}\right]=-\mE\left[D(P_{Y_{k-2}^n|U_{k-2}^n,\cC}\|P_{Y_{k-2}^n|U_{k-2}^n})\right]\leq 0,
\end{flalign}
where $D(P||Q)=\sum_i P(i)\log\frac{P(i)}{Q(i)}$ is the Kullback-Leibler divergence between two distributions $P$ and $Q$.}

{Conditioned on a realization $C$ of the random codebook $\cC$, it follows that
	\begin{flalign}
	&P(y_{k-2}^n|w_{k-1,2},w_{k,1},w_{k,2},\ldots,w_{K,1},w_{K,2},u_{k-2}^n,C)\nn\\
	&=\sum_{w_{k-1,1},l_{k-1}^K}P(y_{k-2}^n, w_{k-1,1},l_{k-1}^K|w_{k-1,2},w_{k,1},w_{k,2},\ldots,w_{K,1},w_{K,2},u_{k-2}^n,C)\nn\\
	&=\frac{1}{2^{n(R_{k-1,1}  +R_{k-1,3}+\ldots+ R_{K,3} )}}\sum_{w_{k-1,1},l_{k-1}^K} P(y_{k-2}^n| w_{k-1,1},w_{k-1,2},\ldots,w_{K,1},w_{K,2},l_{k-1}^K,u_{k-2}^n,C)\nn\\
	&=\frac{1}{2^{n(R_{k-1,1}  +R_{k-1,3}+\ldots+ R_{K,3} )}}\sum_{w_{k-1,1},l_{k-1}^K} P(y_{k-2}^n|u_K^n(w_1,\ldots,w_{k,1},w_{k,2},\ldots,w_{K,1},w_{K,2},l_3^K),C)\nn\\
	&=\frac{1}{2^{n(R_{k-1,1}  +R_{k-1,3}+\ldots+ R_{K,3} )}}\sum_{w_{k-1,1},l_{k-1}^K} P(y_{k-2}^n|u_K^n(w_1,\ldots,w_{k,1},w_{k,2},\ldots,w_{K,1},w_{K,2},l_3^K)),
	\end{flalign}
	where the last step is due to the Markov chain condition $\cC\rightarrow U_K^n\rightarrow Y_{k-2}^n$.
}

{ Due to the symmetry of the random codebook construction,  when computing the expectation in \eqref{eq:ED}, we can assume that all the indices except $(W_{k-1,1},L_{k-1}^K)$ are fixed constants and equal to one. For notational convenience, we only include the indices $(W_{k-1,1},L_{k-1}^K)$ and ignore all those fixed indices when labeling the codewords. For example, instead of $u_{k-2}^n(w_1,\ldots,w_{k-2,1},w_{k-2,2},l_{k-2})$ and $u_K^n(w_1,\ldots,w_{k,1},w_{k,2},\ldots,w_{K,1},w_{K,2},l_3^K)$, we use $u_{k-2}^n$ and $u_K^n(w_{k-1,1},l_{k-1}^K)$.    }

{Following steps similar to those in \cite{hou2013informational}, it can be shown that
	\begin{flalign}\label{eq:ccc}
	&\mE\left[\log\frac{P(Y_{k-2}^n|\mW_k,U_{k-2}^n,\cC)}{P(Y_{k-2}^n|U_{k-2}^n)} \right]\nn\\
	&\overset{(a)}{=}\sum_C P(C)\sum_{y_{k-2}^n}P(y_{k-2}^n|1,\ldots, 1,u_{k-2}^n,C)\log\frac{P(y_{k-2}^n|1,\ldots,1,u_{k-2}^n,C)}{P(y_{k-2}^n|u_{k-2}^n)},\nn\\
	&=\sum_C P(C) \sum_{y_{k-2}^n} \frac{1}{2^{n(R_{k-1,1}  +R_{k-1,3}+\ldots+ R_{K,3} )}}\sum_{w_{k-1,1},l_{k-1}^K} P(y_{k-2}^n|u_K^n(w_{k-1,1}l_{k-1}^K))\nn\\
	&\quad\log\frac{\sum_{\tilde w_{k-1,1},\tilde l_{k-1}^K}P(y_{k-2}^n|u_{K}^n(\tilde w_{k-1,1},\tilde l_{k-1}^K))}{2^{n(R_{k-1,1}  +R_{k-1,3}+\ldots+ R_{K,3} )}P(y_{k-2}^n|u_{k-2}^n)}\nn\\
	&\overset{(b)}{=}\sum_{C} P(u_{k-2}^n)\prod_{\hat w_{k-1,1},\hat l_{k-1}} \bigg[P(u_{k-1}^n(\hat w_{k-1,1},\hat l_{k-1})|u_{k-2}^n)\bigg[ \cdots\prod_{\hat l_K}P(u_K^n(\hat w_{k-1,1},\hat l_{k-1}^K)|u_{K-1}^n(\hat w_{k-1,1},\hat l_{k-1}^{K-1})) \bigg]\bigg]\nn\\
	&\quad\sum_{y_{k-2}^n} \frac{1}{2^{n(R_{k-1,1}  +R_{k-1,3}+\ldots+ R_{K,3} )}}\sum_{w_{k-1,1},l_{k-1}^K} P(y_{k-2}^n|u_K^n(w_{k-1,1}l_{k-1}^K))\nn\\
	&\quad\log\frac{\sum_{\tilde w_{k-1,1},\tilde l_{k-1}^K}P(y_{k-2}^n|u_{K}^n(\tilde w_{k-1,1},\tilde l_{k-1}^K))}{2^{n(R_{k-1,1}  +R_{k-1,3}+\ldots+ R_{K,3} )}P(y_{k-2}^n|u_{k-2}^n)}\nn\\
	&\overset{(c)}{=} \frac{1}{2^{n(R_{k-1,1}  +R_{k-1,3}+\ldots+ R_{K,3} )}}\sum_{C}\sum_{y_{k-2}^n}\sum_{w_{k-1,1},l_{k-1}^K}  P(u_{k-2}^n)\prod_{\hat w_{k-1,1},\hat l_{k-1}} \bigg[P(u_{k-1}^n(\hat w_{k-1,1},\hat l_{k-1})|u_{k-2}^n)\bigg[\cdots \nn\\
	&\quad\prod_{\hat l_{K}}P(u_K^n(\hat w_{k-1,1},\hat l_{k-1}^K)|u_{K-1}^n(\hat w_{k-1,1},\hat l_{k-1}^{K-1})) \bigg]\bigg]P(y_{k-2}^n|u_K^n(w_{k-1,1}l_{k-1}^K))\nn\\
	&\quad\log\frac{\sum_{\tilde w_{k-1,1},\tilde l_{k-1}^K}P(y_{k-2}^n|u_{K}^n(\tilde w_{k-1,1},\tilde l_{k-1}^K))}{2^{n(R_{k-1,1}  +R_{k-1,3}+\ldots+ R_{K,3} )}P(y_{k-2}^n|u_{k-2}^n)},
	\end{flalign}
	where in (a), by the symmetry of the random codebook construction, we let $U_{k-2}^n=u_{k-2}^n(1,\ldots,1)$, $\mW_k=(1,\ldots,1)$; and  in (b) and the following equations, $C$ consists only of those codewords with all the indices except $(w_{k-1,1},l_{k-1}^K)$ being one; (c) is obtained by reordering those summations. 	
}%\sum_{\tilde w_{k-1,1},\tilde l_{k-1}}\sum_{u_{k-1}^n(\tilde w_{k-1,1},\tilde l_{k-1})}\cdots \sum_{\tilde l_K}\sum_{u_K^n(\tilde w_{k-1},\tilde l_{k-1}^K)}

{From \eqref{eq:ccc}, it follows that
	\begin{flalign}\label{eq:expecation}
	&\mE\left[\log\frac{P(Y_{k-2}^n|\mW_k,U_{k-2}^n,\cC)}{P(Y_{k-2}^n|U_{k-2}^n)} \right]\nn\\
	&=\frac{1}{2^{n(R_{k-1,1}  +R_{k-1,3}+\ldots+ R_{K,3} )}}\sum_{w_{k-1,1},l_{k-1}^K}\sum_{y_{k-2}^n}\sum_{u_{k-2}^n} \sum_{u_{k-1}^n(w_{k-1,1},l_{k-1})}\cdots\sum_{u_{K}^n(w_{k-1,1}l_{k-1}^K)} \nn\\
	&\quad P(y_{k-2}^n,u_{K}^n(w_{k-1,1}l_{k-1}^K),\cdots,u_{k-2}^n)\nn\\
	&\quad   \sum_{C\setminus \{u_{k-2}^n,\cdots,u_K^n(w_{k-1,1}l_{k-1}^K)\}}\prod_{\hat w_{k-1,1},\hat l_{k-1}} \bigg[P(u_{k-1}^n(\hat w_{k-1,1},\hat l_{k-1})|u_{k-2}^n)\bigg[\cdots\nn\\
	&\quad \prod_{\hat l_{K}}P(u_K^n(\hat w_{k-1,1},\hat l_{k-1}^K)|u_{K-1}^n(\hat w_{k-1,1},\hat l_{k-1}^{K-1})) \bigg]\bigg]\frac{1}{P(u_{k-1}^n(w_{k-1},l_{k-1}),\ldots, w_K(w_{k-1},l_{k-1}^K)|u_{k-2}^n)}\nn\\
	&\quad\log\frac{\sum_{\tilde w_{k-1,1},\tilde l_{k-1}^K}P(y_{k-2}^n|u_{K}^n(\tilde w_{k-1,1},\tilde l_{k-1}^K))}{2^{n(R_{k-1,1}  +R_{k-1,3}+\ldots+ R_{K,3} )}P(y_{k-2}^n|u_{k-2}^n)}\nn\\
	&\overset{(a)}{\leq } \frac{1}{2^{n(R_{k-1,1}  +R_{k-1,3}+\ldots+ R_{K,3} )}}\sum_{w_{k-1,1},l_{k-1}^K}\sum_{y_{k-2}^n}\sum_{u_{k-2}^n} \sum_{u_{k-1}^n(w_{k-1,1},l_{k-1})}\cdots\sum_{u_{K}^n(w_{k-1,1}l_{k-1}^K)} \nn\\
	&\quad P(y_{k-2}^n,u_{K}^n(w_{k-1,1}l_{k-1}^K),\cdots,u_{k-2}^n)\log \Bigg(\mE \frac{\sum_{(\tilde w_{k-1,1},\tilde l_{k-1}^K)\neq (w_{k-1,1},l_{k-1}^K)}P(y_{k-2}^n|U_{K}^n(\tilde w_{k-1,1},\tilde l_{k-1}^K))}{2^{n(R_{k-1,1}  +R_{k-1,3}+\ldots+ R_{K,3} )}P(y_{k-2}^n|u_{k-2}^n)}\nn\\
	&\quad + \frac{P(y_{k-2}^n|u_{K}^n( w_{k-1,1}, l_{k-1}^K))}{2^{n(R_{k-1,1}  +R_{k-1,3}+\ldots+ R_{K,3} )}P(y_{k-2}^n|u_{k-2}^n)}\Bigg),
	\end{flalign}
	where $(a)$ follows by the concavity of the logarithm and Jensen's inequality applied to the expectation over all the codewords except $\left(u_{k-2}^n,u_{k-1}^n(w_{k-1},l_{k-1}),\ldots, u_K^n(w_{k-1},l_{k-1}^K)\right)$.
}

{We now consider the expectation in \eqref{eq:expecation} for different values of $(\tilde w_{k-1,1},\tilde l_{k-1}^K)$.
We first define
\begin{flalign}\label{eq:28}
\frac{P(y_{k-2}^n|u_{K}^n( w_{k-1,1}, l_{k-1}^K))}{2^{n(R_{k-1,1}  +R_{k-1,3}+\ldots+ R_{K,3} )}P(y_{k-2}^n|u_{k-2}^n)} \overset{\Delta}{=}A_K.
\end{flalign}
For $(\tilde w_{k-1,1},\tilde l_{k-1}^{K-1})=( w_{k-1,1}, l_{k-1}^{K-1})$ but $\tilde l_{K}\neq l_K$, we obtain a term
\begin{flalign}
&\sum_{\tilde l_{K}\neq l_K}\frac{P(y_{k-2}^n|u_{K-1}^n( w_{k-1,1}, l_{k-1}^{K-1}))}{2^{n(R_{k-1,1}  +R_{k-1,3}+\ldots+ R_{K,3} )}P(y_{k-2}^n|u_{k-2}^n)}\nn\\
&\leq \frac{P(y_{k-2}^n|u_{K-1}^n( w_{k-1,1}, l_{k-1}^{K-1}))}{2^{n(R_{k-1,1}  +R_{k-1,3}+\ldots+ R_{K-1,3} )}P(y_{k-2}^n|u_{k-2}^n)}\nn\\
&\overset{\Delta}{=}A_{K-1}.
\end{flalign}
More generally, for any $k-1\leq j\leq K-1$, for $(\tilde w_{k-1,1},\tilde l_{k-1}^{j})=( w_{k-1,1}, l_{k-1}^{j})$ but $\tilde l_{j+1}\neq l_{j+1}$, we obtain a term
\begin{flalign}\label{eq:31}
&\sum_{\tilde l_{j+1}^K: \tilde l_{j+1}\neq l_{j+1}}\frac{P(y_{k-2}^n|u_{j}^n( w_{k-1,1}, l_{k-1}^{j}))}{2^{n(R_{k-1,1}  +R_{k-1,3}+\ldots+ R_{K,3} )}P(y_{k-2}^n|u_{k-2}^n)}\nn\\
&\leq \frac{P(y_{k-2}^n|u_{j}^n( w_{k-1,1}, l_{k-1}^{j}))}{2^{n(R_{k-1,1}  +R_{k-1,3}+\ldots+ R_{j,3} )}P(y_{k-2}^n|u_{k-2}^n)}\nn\\
&\overset{\Delta}{=}A_{j}.
\end{flalign}
For $(\tilde w_{k-1,1},\tilde l_{k-1})\neq ( w_{k-1,1}, l_{k-1})$, we obtain a term
\begin{flalign}\label{eq:32}
&\sum_{(\tilde w_{k-1,1}, \tilde l_{j+1}^K): (\tilde w_{k-1,1},\tilde l_{k-1})\neq ( w_{k-1,1}, l_{k-1})}\frac{P(y_{k-2}^n|u_{k-2}^n)}{2^{n(R_{k-1,1}  +R_{k-1,3}+\ldots+ R_{K,3} )}P(y_{k-2}^n|u_{k-2}^n)} \leq 1.
\end{flalign}
}

{Combining \eqref{eq:28} \eqref{eq:31} and \eqref{eq:32} yields that the term within the $\log$ in \eqref{eq:expecation} is upper bounded by 
$ 1+\sum_{j=k-1}^{K}A_j,
$
	which further implies that
	\begin{flalign}\label{eq:qwe}
	\mE\left[\log\frac{P(Y_{k-2}^n|\mW_k,U_{k-2}^n,\cC)}{P(Y_{k-2}^n|U_{k-2}^n)} \right]&\leq \mE \log\left(1+\sum_{j=k-1}^{K}\frac{P(Y_{k-2}^n|U_{j}^n( W_{k-1,1}, L_{k-1}^{j}))}{2^{n(R_{k-1,1}  +R_{k-1,3}+\ldots+ R_{j,3} )}P(Y_{k-2}^n|U_{k-2}^n)}\right)\nn\\
	&\leq \sum_{j=k-1}^{K}\mE \log\left(1+\frac{P(Y_{k-2}^n|U_{j}^n( W_{k-1,1}, L_{k-1}^{j}))}{2^{n(R_{k-1,1}  +R_{k-1,3}+\ldots+ R_{j,3} )}P(Y_{k-2}^n|U_{k-2}^n)}\right).
	\end{flalign}
	By the symmetry of the random codeword generation, we assume that $(W_{k-1,1},L_{k-1}^K)$ are fixed, and thus in the following proof, we ignore these indices.
	For any $k-1\leq j\leq K$, it then follows that
	\begin{flalign}
	&\mE\log\left(1+\frac{P(Y_{k-2}^n|U_{j}^n)}{2^{n(R_{k-1,1}  +R_{k-1,3}+\ldots+ R_{j,3} )}P(Y_{k-2}^n|U_{k-2}^n)}\right)\nn\\
	&=\sum_{\substack{(u_{k-2}^n,u_j^n,y_{k-2}^n) \\\in T_\epsilon^n(P_{U_{k-2}U_jY_{k-2}})}}P(u_{k-2}^n,u_j^n,y_{k-2}^n)\log\left(1+\frac{P(y_{k-2}^n|u_{j}^n)}{2^{n(R_{k-1,1}  +R_{k-1,3}+\ldots+ R_{j,3} )}P(y_{k-2}^n|u_{k-2}^n)}\right)\nn\\
	&\quad + \sum_{\substack{(u_{k-2}^n,u_j^n,y_{k-2}^n)\\\notin T_\epsilon^n(P_{U_{k-2}U_jY_{k-2}})}}P(u_{k-2}^n,u_j^n,y_{k-2}^n)\log\left(1+\frac{P(y_{k-2}^n|u_{j}^n)}{2^{n(R_{k-1,1}  +R_{k-1,3}+\ldots+ R_{j,3} )}P(y_{k-2}^n|u_{k-2}^n)}\right)\nn\\
	&\overset{\Delta}{=}d_1+d_2.
	\end{flalign}
	Using the inequalities in \cite[Appendix]{orlitsky1995coding} and following steps similar to those in \cite{hou2013informational}, we have
	\begin{flalign}
	&d_1\leq \sum_{\substack{(u_{k-2}^n,u_j^n,y_{k-2}^n) \\  \in T_\epsilon^n(P_{U_{k-2}U_jY_{k-2}})}}P(u_{k-2}^n,u_j^n,y_{k-2}^n)\log\left(1+\frac{2^{-n(1-\epsilon)H(Y_{k-2}|U_j)}}{2^{n(R_{k-1,1}  +R_{k-1,3}+\ldots+ R_{j,3} )}2^{-n(1+\epsilon)H(Y_{k-2}|U_{k-2})}}\right)\nn\\
	&\leq \log\left(1+\frac{2^{-n(1-\epsilon)H(Y_{k-2}|U_j)}}{2^{n(R_{k-1,1}  +R_{k-1,3}+\ldots+ R_{j,3} )}2^{-n(1+\epsilon)H(Y_{k-2}|U_{k-2})}}\right),
	\end{flalign}
	which vanishes as $n\rightarrow \infty$ if
	\begin{flalign}
	R_{k-1,1}  +R_{k-1,3}+\ldots+ R_{j,3} >I(U_j;Y_{k-2}|U_{k-2})+2\epsilon H(Y_{k-2}|U_{k-2}).
	\end{flalign}
	To show $d_2\rightarrow 0$ as $n\rightarrow \infty$, it follows that
	\begin{flalign}
	&d_2\leq \sum_{\substack{(u_{k-2}^n,u_j^n,y_{k-2}^n)\notin T_\epsilon^n(P_{U_{k-2}U_jY_{k-2}})\\ (u_{k-2}^n,u_j^n,y_{k-2}^n) \in \text{supp}(P_{U_{k-2}^n,U_j^n,Y_{k-2}^n}) }}P(u_{k-2}^n,u_j^n,y_{k-2}^n)\log\left(1+\left(\frac{1}{\mu}\right)^n\right)\nn\\
	&\leq 2|\mathcal U_{k-2}||\mathcal U_j||\mathcal Y_{k-2}|e^{-\epsilon^2\phi n/3}n\log\left(1+\frac{1}{\mu}\right)\nn\\
	&\rightarrow 0, \text{ as } n\rightarrow \infty.
	\end{flalign}
	where supp$(P_X)$ is defined to be the support of a distribution $P_X$, $|\mathcal U_{k-2}|,|\mathcal U_j|$ and $|\mathcal Y_{k-2}|$ are the support sizes of $U_{k-2}$, $U_j$ and $Y_{k-2}$, respectively, and
	\begin{flalign}
		\mu&=\min_{(u_{k-2},y_{k-2})\in \text{supp}(P_{U_{k-2}Y_{k-2}})} P(y_{k-2}|u_{k-2}),\nn\\
		\phi&=\min_{(u_{k-2},u_j,y_{k-2})\in \text{supp}(P_{U_{k-2}U_jY_{k-2}})} P(u_{k-2}u_jy_{k-2}).
	\end{flalign}
}

{Therefore, if the following conditions are satisfied for $3\leq k\leq K$ and $k-1\leq j\leq K$:
	\begin{flalign}\label{eq:con:secrecy}
	R_{k-1,1}  +R_{k-1,3}+\ldots+ R_{j,3} >I(U_j;Y_{k-2}|U_{k-2}),
	\end{flalign}
	then
	\begin{flalign}
	I(\mW_k;Y_{k-2}^n|U_{k-2}^n,\cC)\rightarrow 0, \text{ as } n\rightarrow \infty, \text { for }3\leq k\leq K.
	\end{flalign}
}

Combining the bounds in \eqref{eq:decoding} and \eqref{eq:con:secrecy}, and by choosing $R_{k,1}+R_{k,2}+R_{k,3}=I(U_k;Y_k|U_{k-1})$, we conclude that the rate tuple $(R_1,R_{2,1}, R_{2,2}, \ldots,R_{K,1},R_{K,2})$ is achievable if
\begin{flalign}
  R_1&\leq I(U_1;Y_1),\nn\\
  R_{k,1}+R_{k,2}&\leq I(U_k;Y_k|U_{k-1})  \text{, for } 2\leq k\leq K,\nn\\
  R_{k-1,2}+\sum_{i=k}^j (R_{i,1}+R_{i,2})&\leq \sum_{i=k-1}^j I(U_i;Y_i|U_{i-1})-I(U_j;Y_{k-2}|U_{k-2}), \nn\\
   &\hspace{2cm}\text{for } 3\leq k\leq K, \text{ and } k-1\leq j\leq K.
\end{flalign}

\emph{Rate Sharing:} We note that our achievable scheme guarantees $W_{k-1,2}, W_{k,1},W_{k,2},\ldots,W_{K,1},\\W_{K,2}$ to be secure from receiver $Y_{k-2}$, for $3\leq k\leq K$.  Furthermore, due to the degradedness condition, $W_{k-1,2}$ can be decoded by receiver $Y_k$. Thus, $W_{k-1,2}$ satisfies both the decoding and secrecy requirements as $W_k$. Hence, the rate of $W_{k-1,2}$ can be counted towards either $R_{k-1}$ or $R_k$. Based on such an understanding, we design the following rate sharing scheme. We define $R_2=R_{2,1}$, $R_k=R_{k-1,2}+R_{k,1}$ for $3\leq k\leq K-1$, and $R_K=R_{K-1,2}+R_{K,1}+R_{K,2}$, and include these equations to the above achievable region.
We then perform Fourier-Motzkin elimination to eliminate $R_{k,1},R_{k,2}$ for $2\leq k\leq K$ and obtain a closed-form achievable rate region. Such a process involves eliminating $2K-2$ variables from the order of  $K^2$ bounds, which is intractable for arbitrary $K$. We propose an inductive Fourier Motzkin elimination approach as shown in Appendix \ref{app:Fourier}, and obtain the achievable region given in Theorem \ref{thm:capacity}.

\section{Inductive Fourier-Motzkin Elimination}\label{app:Fourier}
As we have shown in Appendix \ref{app:achiv}, we need to eliminate $R_{k,1},R_{k,2}$ for $2\leq k\leq K$ in the following region:
\begin{subequations}
\begin{flalign}
  R_1&\leq I(U_1;Y_1),\label{eq:sub1a}\\
  R_{k,1}+R_{k,2}&\leq I(U_k;Y_k|U_{k-1})  \text{, for } 2\leq k\leq K,\label{eq:sub1b}\\
  R_{l-1,2}+\sum_{i=l}^j (R_{i,1}+R_{i,2})&\leq \sum_{i=l-1}^j I(U_i;Y_i|U_{i-1})-I(U_j;Y_{l-2}|U_{l-2}),\label{eq:sub1c}\\
   &\hspace{2cm}\text{for } 3\leq l\leq K , l-1\leq j\leq K,\nn\\
  R_2&=R_{2,1},\label{eq:sub2a}\\
  R_k&=R_{k-1,2}+R_{k,1}, \text{ for } 3\leq k\leq K-1,\label{eq:sub2b}\\
  R_K&=R_{K-1,2}+R_{K,1}+R_{K,2},\label{eq:sub2c}
\end{flalign}
\end{subequations}
where the bounds \eqref{eq:sub1a}, \eqref{eq:sub1b} and \eqref{eq:sub1c} correspond to the achievable region after rate splitting, which are expressed in terms of component rates, and the bounds \eqref{eq:sub2a}, \eqref{eq:sub2b} and \eqref{eq:sub2c} are corresponding to the rate sharing strategy.

It can be seen that the total number of bounds in the above region is on the order of $K^2$ over which  $2K-2$ variables need to be eliminated. Directly applying Fourier-Motzkin elimination is not analytically tractable. We design an inductive algorithm, in which we eliminate the rate pairs $(R_{k-1,2}, R_{k,1})$ for $3\leq k\leq K-1$ one at each step, and finally eliminate $(R_{K-1,2},R_{K,1},R_{K,2})$. We first replace $R_{2,1}$ with $R_2$, $R_{k-1,2}+R_{k,1}$ with $R_k$  for $3\leq k\leq K-1$, and $R_{K-1,2}+R_{K,1}+R_{K,2}$ with $R_K$, and we obtain the following region:
\begin{flalign}\label{eq:region1}
  R_1&\leq I(U_1;Y_1),\nn\\
  R_2+R_{2,2}&\leq I(U_2;Y_2|U_{1})\nn\\
  R_{k,1}+R_{k,2}&\leq I(U_k;Y_k|U_{k-1})  \text{, for } 3\leq k\leq K,\nn\\
  \sum_{i=l}^{j}R_i +R_{j,2}&\leq \sum_{i=l-1}^j I(U_i;Y_i|U_{i-1})-I(U_j;Y_{l-2}|U_{l-2}), \nn\\
   &\hspace{2cm}\text{for } 3\leq l\leq K , l-1\leq j\leq K-1,\nn\\
  \sum_{i=l}^K R_i &\leq \sum_{i=l-1}^K I(U_i;Y_i|U_{i-1})-I(U_K;Y_{l-2}|U_{l-2}),\nn\\
  &\hspace{2cm}\text{ for } 3\leq l\leq K,\nn\\
    R_k&=R_{k-1,2}+R_{k,1}, \text{ for } 3\leq k\leq K-1,\nn\\
  R_K&=R_{K-1,2}+R_{K,1}+R_{K,2}.
\end{flalign}

To start the elimination process, we first eliminate $(R_{2,2},R_{3,1})$ from the inequalities given below, corresponding to the decoding and secrecy requirements of receiver 1 to receiver 3:
\begin{flalign}
  R_1&\leq I(U_1;Y_1),\nn\\
  R_2+R_{2,2}&\leq I(U_2;Y_2|U_{1}),\nn\\
  R_{3,1}+R_{3,2}&\leq I(U_3;Y_3|U_{2}),\nn\\
  R_{2,2}&\leq I(U_2;Y_2|U_1)-I(U_2;Y_1|U_1),\nn\\
  R_3+R_{3,2}&\leq \sum_{i=2}^3 I(U_i;Y_i|U_{i-1})-I(U_3;Y_{1}|U_{1}).\nn\\
  R_{3,2}&\leq I(U_3;Y_3|U_{2})-I(U_3;Y_2|U_2),\nn\\
  R_3&=R_{2,2}+R_{3,1}.
\end{flalign}
We then obtain the following inequalities after elimination:
\begin{flalign}\label{eq:s3}
  R_1&\leq I(U_1;Y_1),\nn\\
  R_2&\leq I(U_2;Y_2|U_{1}),\nn\\
  \sum_{i=2}^3 R_i +R_{3,2}&\leq \sum_{i=2}^3 I(U_i;Y_i|U_{i-1}),\nn\\
  R_3+R_{3,2}&\leq \sum_{i=2}^3 I(U_i;Y_i|U_{i-1})-I(U_3;Y_1|U_1),\nn\\
  R_{3,2}&\leq I(U_3;Y_3|U_{2})-I(U_3;Y_2|U_2),
\end{flalign}
which we denote as $\mathcal{R}_3$.

We then eliminate $(R_{3,2},R_{4,1})$ from the inequalities in $\mathcal{R}_3$ and the inequalities given below, which together are corresponding to the decoding and secrecy requirements of receiver 1 to receiver 4:
\begin{flalign}
  R_{4,1}+R_{4,2}&\leq I(U_4;Y_4|U_{3}),\nn\\
  \sum_{i=j}^4 R_i +R_{4,2}&\leq \sum_{i=j-1}^4 I(U_i;Y_i|U_{i-1})-I(U_4;Y_{j-2}|U_{j-2}), \text{ for } 3\leq j\leq 5\nn\\
  R_4&=R_{3,2}+R_{4,1}.
  % R_4+R_{4,2}&\leq \sum_{i=3}^4 I(U_i;Y_i|U_{i-1})-I(U_4;Y_2|U_2),
\end{flalign}
We then obtain the following bounds after elimination:
\begin{flalign}\label{eq:s4}
    R_1&\leq I(U_1;Y_1),\nn\\
  \sum_{i=2}^j R_i&\leq \sum_{i=2}^j I(U_i;Y_i|U_{i-1}), \text{ for } 2\leq j\leq 3,\nn\\
  \sum_{i=2}^4 R_i+R_{4,2}&\leq \sum_{i=2}^4 I(U_i;Y_i|U_{i-1}),\nn\\
  \sum_{i=l}^j R_i&\leq \sum_{i=l-1}^j I(U_i;Y_i|U_{i-1})-I(U_j;Y_{l-2}|U_{l-2}),\nn\\
  &\hspace{3cm}\text{for } 3\leq l\leq j\leq 3,\nn\\
  \sum_{i=l}^4 R_i+R_{4,2}&\leq \sum_{i=l-1}^4 I(U_i;Y_i|U_{i-1})-I(U_4;Y_{l-2}|U_{l-2}),\nn\\
  &\hspace{3cm}\text{ for } 3\leq l\leq 5,
\end{flalign}
which we denote as $\mathcal{R}_4$.

%We then eliminate $(R_{4,2},R_{5,1})$ among the inequalities in \eqref{eq:s4} and the following inequalities that are corresponding to the decoding and secrecy requirements of receiver 1 to receiver 5:
%\begin{flalign}
%   R_{5,1}+R_{5,2}&\leq I(U_5;Y_5|U_{4}),\nn\\
%  \sum_{i=j}^5 R_i +R_{5,2}&\leq \sum_{i=j-1}^5 I(U_i;Y_i|U_{i-1})-I(U_5;Y_{j-2}|U_{j-2}), \text{ for } 3\leq j\leq 6\nn\\
%  R_5&=R_{4,2}+R_{5,1}.
%\end{flalign}
%We obtain the following region:
%\begin{flalign}\label{eq:s5}
%    R_1&\leq I(U_1;Y_1),\nn\\
%  \sum_{i=2}^j R_i&\leq \sum_{i=2}^j I(U_i;Y_i|U_{i-1}), \text{ for } 2\leq j\leq 4,\nn\\
%  \sum_{i=2}^5 R_i+R_{5,2}&\leq \sum_{i=2}^5 I(U_i;Y_i|U_{i-1}),\nn\\
%  \sum_{i=l}^j R_i&\leq \sum_{i=l-1}^j I(U_i;Y_i|U_{i-1})-I(U_j;Y_{l-2}|U_{l-2}),\nn\\
%  &\hspace{3cm}\text{for } 3\leq l\leq j\leq 4,\nn\\
%  \sum_{i=l}^5 R_i+R_{5,2}&\leq \sum_{i=l-1}^5 I(U_i;Y_i|U_{i-1})-I(U_4;Y_{l-2}|U_{l-2}),\nn\\
%  &\hspace{3cm}\text{ for } 3\leq l\leq 6,
%\end{flalign}
%which we denote as $\mathcal{R}_5$.

As we observe, the region $\mathcal R_3$ and $\mathcal R_4$ conform to the following structure for $k=3$ and $k=4$:
\begin{flalign}\label{eq:structure}
  R_1&\leq I(U_1;Y_1),\nn\\
  \sum_{i=2}^j R_i&\leq \sum_{i=2}^j I(U_i;Y_i|U_{i-1}), \text{ for } 2\leq j\leq k-1,\nn\\
  \sum_{i=2}^k R_i+R_{k,2}&\leq \sum_{i=2}^k I(U_i;Y_i|U_{i-1}),\nn\\
  \sum_{i=l}^j R_i&\leq \sum_{i=l-1}^j I(U_i;Y_i|U_{i-1})-I(U_j;Y_{l-2}|U_{l-2}),\nn\\
  &\hspace{3cm}\text{for } 3\leq l\leq j\leq k-1,\nn\\
  \sum_{i=l}^k R_i+R_{k,2}&\leq \sum_{i=l-1}^k I(U_i;Y_i|U_{i-1})-I(U_k;Y_{l-2}|U_{l-2}),\nn\\
  &\hspace{3cm}\text{ for } 3\leq l\leq k+1.
\end{flalign}

We next show that the region $\mathcal R_k$ takes the structure \eqref{eq:structure} for any $3\leq k\leq K-1$ using induction. We have verified such a claim for $k=3,4$. If such a claim holds for $\mathcal R_k$, we eliminate $R_{k,2}$ and $R_{k+1,1}$ from the inequalities in $\mathcal R_k$ and the inequalities given below, which together are corresponding to the decoding and secrecy requirements of receiver 1 to receiver $k+1$:
\begin{flalign}
  R_{k+1,1}+R_{k+1,2}&\leq I(U_{k+1};Y_{k+1}|U_{k}),\nn\\
  \sum_{i=j}^{k+1} R_i +R_{k+1,2}&\leq \sum_{i=j-1}^{k+1} I(U_i;Y_i|U_{i-1})-I(U_{k+1};Y_{j-2}|U_{j-2}), \text{ for } 3\leq j\leq k+2\nn\\
  R_{k+1}&=R_{k,2}+R_{k+1,1}.
  % R_4+R_{4,2}&\leq \sum_{i=3}^4 I(U_i;Y_i|U_{i-1})-I(U_4;Y_2|U_2),
\end{flalign}
Then the resulting region, following standard steps of Fourier-Motzkin elimination to eliminate $R_{k,2}$ and $R_{k+1,1}$, equals \eqref{eq:structure} for $k+1$.

Finally, we eliminate $(R_{K-1,2},R_{K,1},R_{K,2})$, and obtain the achievable region in Theorem \ref{thm:capacity}.

\section{Converse Proof of Theorem \ref{thm:capacity}}\label{app:converse}

{We note that the converse proof is based on the weak secrecy requirement, which is necessarily valid under the strong secrecy requirement. Such a converse proof also implies that the secrecy capacity region under the weak and strong secrecy requirements are the same.}

By Fano's inequality and the secrecy requirements, we have the following inequalities:
\begin{flalign}
H(W_k|Y_k^n)\leq n\epsilon_n&,\quad\text{for }1\leq k\leq K,\\
I(W_{k},\ldots,W_K;Y_{k-2}^n)\leq \epsilon_n\leq n\epsilon_n&,\quad \text{for }3\leq k\leq K,
\end{flalign}
both of which implies that
\begin{flalign}
  I(W_{k},\ldots,W_K;Y_{k-2}^n|W_1,\ldots,W_{k-2})\leq n\epsilon_n&,\quad \text{for }3\leq k\leq K.\label{eq:conv1}
\end{flalign}

We denote $Y_k^{i-1}:=(Y_{k,1},\ldots,Y_{k,i-1})$, and $Y_{k,i+1}^n:=(Y_{k,i+1},\ldots,Y_{k,n})$. We set $U_{1,i}:=(W_1,Y_1^{i-1})$, $U_{2,i}:=(W_1,W_2,Y_2^{i-1})$, $U_{k,i}:=(W_1,\ldots,W_k,Y_k^{i-1},Y_{k-2,i+1}^n)$, for $3\leq k\leq K$. We note that $Y_0^n=Y_{-1}^n=\Phi$. Due to the degradedness condition, it can be verified that $(U_{1,i},U_{2,i},\ldots,U_{K-1,i},U_{K,i},X_i)$ satisfy the following Markov chain condition:
\begin{flalign}\label{eq:Markov2}
  U_{1,i}\rightarrow U_{2,i}\rightarrow \ldots\rightarrow U_{K,i}\rightarrow X_i\rightarrow Y_{K,i}\rightarrow\ldots\rightarrow Y_{1,i},\text{ for } 1\leq i\leq n.
\end{flalign}

We first bound the rate $R_1$. Since $W_1$ is only required to be decoded by receiver $Y_1$, we obtain the following bound:
\begin{flalign}
  nR_1&=H(W_1)=I(W_1;Y_1^n)+H(W_1|Y_1^n)\nn\\
  &\overset{(a)}{\leq} I(W_1;Y_1^n)+n\epsilon_n=\sum_{i=1}^nI(W_1;Y_{1i}|Y_{1}^{i-1})+n\epsilon_n\nn\\
  &\leq \sum_{i=1}^nI(W_1,Y_{1}^{i-1};Y_{1i})+n\epsilon_n=\sum_{i=1}^nI(U_{1,i};Y_{1,i})+n\epsilon_n,
\end{flalign}
where $(a)$ is due to Fano's inequality.
%\begin{flalign}
%  \sum_{i=2}^{k}R_i&\leq \sum_{i=1}^{k}I(U_i;Y_i|U_{i-1}),\text{  for }2\leq k\leq K, \\
%  \sum_{i=l}^j R_i &\leq \sum_{i=l-1}^{j}I(U_i;Y_i|U_{i-1})-I(U_j;Y_{l-2}|U_{l-2}), \text{ for } 3\leq l\leq j\leq K,\label{eq:capacity3}
%\end{flalign}

We further bound the rate $R_2$ as follows:
\begin{flalign}
  nR_2&=H(W_2)=H(W_2|W_1)=I(W_2;Y_2^n|W_1)+H(W_2|Y_2^n,W_1)\nn\\
  &\overset{(a)}{\leq} I(W_2;Y_2^n|W_1)+n\epsilon_n\nn\\
  &=\sum_{i=1}^n I(W_2;Y_{2,i}|W_1,Y_2^{i-1})+n\epsilon_n\nn\\
  &\overset{(b)}{\leq} \sum_{i=1}^n I(W_1,W_2,Y_2^{i-1};Y_{2,i}|W_1,Y_1^{i-1})+n\epsilon_n\nn\\
  &=\sum_{i=1}^n I(U_{2,i};Y_{2,i}|U_{1,i})+n\epsilon_n,
\end{flalign}
where $(a)$ is due to Fano's inequality, and $(b)$ is due to the Markov chain condition $Y_1^{i-1}\rightarrow Y_2^{i-1}\rightarrow (W_1,W_2,Y_{2,i})$.

We then bound the sum rate bounds on $\sum_{i=2}^{k}R_i$, for $3\leq k\leq K$: %which are corresponding to the decoding requirements.
\begin{flalign}\label{eq:11}
  n&\sum_{j=2}^{k}R_j=H(W_2,\ldots,W_k)\nn\\
  &\overset{(a)}{=}H(W_2|W_1)+H(W_3|W_1,W_2)+\ldots+H(W_k|W_1,\ldots,W_{k-1})\nn\\
  &\overset{(b)}{\leq} I(W_2;Y_2^n|W_1)+I(W_3;Y_3^n|W_1,W_2)+\ldots+I(W_k;Y_k^n|W_1,\ldots,W_{k-1})+n(k-1)\epsilon_n\nn\\
  &=\sum_{i=1}^n I(W_2;Y_{2,i}|W_1,Y_2^{i-1})+I(W_3;Y_{3,i}|W_1,W_2,Y_3^{i-1})\nn\\
  &\hspace{2cm}+\ldots+I(W_k;Y_{k,i}|W_1,\ldots,W_{k-1},Y_k^{i-1})+n(k-1)\epsilon_n\nn\\
  &=n(k-1)\epsilon_n+\sum_{i=1}^n \Bigg(I(W_2,Y_2^{i-1};Y_{2,i}|W_1,Y_1^{i-1})-I(Y_2^{i-1};Y_{2,i}|W_1,Y_1^{i-1})\nn\\
  &\quad +I(W_3,Y_3^{i-1},Y_{1,i+1}^n;Y_{3,i}|W_1,W_2,Y_2^{i-1})-I(Y_3^{i-1};Y_{3,i}|W_1,W_2,Y_2^{i-1}) \nn\\
  &\quad- I(Y_{1,i+1}^n;Y_{3,i}|W_1,W_2,W_3,Y_3^{i-1})\nn\\
  &\quad+\sum_{j=4}^k \bigg(I(W_j,Y_j^{i-1},Y_{j-2,i+1}^n;Y_{j,i}|W_1,\ldots,W_{j-1},Y_{j-3,i+1}^n,Y_{j-1}^{i-1}) \nn\\
  &\quad+ I(Y_{j-3,i+1}^n;Y_{j,i}|W_1,\ldots,W_{j-1},Y_j^{i-1})\nn\\
  &\quad-I(Y_j^{i-1};Y_{j,i}|W_1,\ldots,W_{j-1},Y_{j-3,i+1}^n,Y_{j-1}^{i-1})-I(Y_{j-2,i+1}^n;Y_{j,i}|W_1,\ldots,W_j,Y_j^{i-1})\bigg)\Bigg)\nn\\
  &\overset{(c)}{\leq} n(k-1)\epsilon_n+ \sum_{j=2}^k\sum_{i=1}^n I(U_{j,i};Y_{j,i}|U_{j-1,i}),
\end{flalign}
where $(a)$ is due to the independence between the messages $(W_1,\ldots,W_k)$, $(b)$ is due to Fano's inequality, and $(c)$ is due to the facts that $-I(Y_2^{i-1};Y_{2,i}|W_1,Y_1^{i-1})\leq 0$, $-I(Y_3^{i-1};Y_{3,i}|W_1,W_2,Y_2^{i-1})\leq 0$, $-I(Y_{k-2,i+1}^n;Y_{k,i}|W_1,\ldots,W_k,Y_k^{i-1})\leq 0$ and the following inequalities:
\begin{flalign}
-I&(Y_{j-2,i+1}^n;Y_{j,i}|W_1,\ldots,W_j,Y_j^{i-1})+ I(Y_{j-2,i+1}^n;Y_{j+1,i}|W_1,\ldots,W_{j},Y_{j+1}^{i-1})\nn\\
-I&(Y_{j+1}^{i-1};Y_{j+1,i}|W_1,\ldots,W_{j},Y_{j-2,i+1}^n,Y_{j}^{i-1})\nn\\
  &\overset{(a)}{=}-I(Y_j^{i-1};Y_{j-2,i}|W_1,\ldots,W_j,Y_{j-2,i+1}^n)+ I(Y_{j+1}^{i-1};Y_{j-2,i}|W_1,\ldots,W_{j},Y_{j-2,i+1}^n)\nn\\
  &\quad -I(Y_{j+1}^{i-1};Y_{j+1,i}|W_1,\ldots,W_{j},Y_{j-2,i+1}^n,Y_{j}^{i-1})\nn\\
  &\overset{(b)}=I(Y_{j+1}^{i-1};Y_{j-2,i}|W_1,\ldots,W_{j},Y_{j-2,i+1}^n,Y_{j}^{i-1})-I(Y_{j+1}^{i-1};Y_{j+1,i}|W_1,\ldots,W_{j},Y_{j-2,i+1}^n,Y_{j}^{i-1})\nn\\
  &\overset{(c)}=-I(Y_{j+1}^{i-1};Y_{j+1,i}|W_1,\ldots,W_{j},Y_{j-2,i+1}^n,Y_{j}^{i-1},Y_{j-2,i})\nn\\
  &\leq 0,
\end{flalign}
where $(a)$ is due to Csisz\'ar's sum identity property \cite{Csiszar1978}, and $(b)$ and $(c)$ are due to the degradedness condition \eqref{eq:Markov}.

We next bound the sum rate bounds on $\sum_{j=l}^k R_j$, for $3\leq l\leq k\leq K$, which correspond to the secrecy constraints:
\begin{flalign}\label{eq:444}
  n\sum_{j=l}^k R_j&=H(W_l,\ldots,W_k)+H(W_{l-1})-H(W_{l-1})\nn\\
  &\overset{(a)}{\leq} \sum_{j=l-1}^k H(W_{j})-H(W_{l-1})+n\epsilon_n-I(W_l\ldots,W_k;Y_{l-2}^n|W_1,\ldots,W_{l-2})\nn\\
  &\overset{(b)}{\leq} \sum_{j=l-1}^k H(W_{j})+n\epsilon_n-I(W_{l-1}\ldots,W_k;Y_{l-2}^n|W_1,\ldots,W_{l-2})
\end{flalign}
where $(a)$ is due to the secrecy requirement \eqref{eq:conv1} and the independence of the messages, and $(b)$ is due to the following fact:
\begin{flalign}
  -H(W_{l-1})&-I(W_l\ldots,W_k;Y_{l-2}^n|W_1,\ldots,W_{l-2})\nn\\
  &=-H(W_{l-1})-H(W_l\ldots,W_k|W_1,\ldots,W_{l-2})+H(W_l\ldots,W_k|Y_{l-2}^n,W_1,\ldots,W_{l-2})\nn\\
  &\overset{(a)}=-H(W_{l-1}\ldots,W_k|W_1,\ldots,W_{l-2})+H(W_l\ldots,W_k|Y_{l-2}^n,W_1,\ldots,W_{l-2})\nn\\
  &\leq -H(W_{l-1}\ldots,W_k|W_1,\ldots,W_{l-2})+H(W_{l-1},W_l\ldots,W_k|Y_{l-2}^n,W_1,\ldots,W_{l-2})\nn\\
  &=-I(W_{l-1}\ldots,W_k;Y_{l-2}^n|W_1,\ldots,W_{l-2}),
\end{flalign}
where $(a)$ is due to the independence of the messages.

We next bound each term in \eqref{eq:444} one by one.
We first bound $H(W_j)$ for $l\leq j\leq k$ as follows:
\begin{flalign}\label{eq:555}
  H(W_j)&\overset{(a)}{\leq} H(W_j|W_1,\ldots,W_{j-1})+n\epsilon_n-H(W_j|Y_j^n,W_1,\ldots,W_{j-1})\nn\\
  &=I(W_j;Y_j^n|W_1,\ldots,W_{j-1})+n\epsilon_n\nn\\
  &=n\epsilon_n+\sum_{i=1}^n I(W_j;Y_{j,i}|W_1,\ldots,W_{j-1},Y_{j}^{i-1})\nn\\
  &=n\epsilon_n+\sum_{i=1}^n \Big(I(W_j,Y_j^{i-1},Y_{j-2,i+1}^n;Y_{j,i}|W_1,\ldots,W_{j-1},Y_{j-1}^{i-1},Y_{j-3,i+1}^n)\nn\\
  &\quad -I(Y_j^{i-1};Y_{j,i}|W_1,\ldots,W_{j-1},Y_{j-1}^{i-1},Y_{j-3,i+1}^n) + I(Y_{j-3,i+1}^n;Y_{j,i}|W_1,\ldots,W_{j-1},Y_j^{i-1})\nn\\
  &\quad -I(Y_{j-2,i+1}^n;Y_{j,i}|W_1,\ldots,W_j,Y_j^{i-1})\Big)\nn\\
  &=n\epsilon_n+\sum_{i=1}^n \Big(I(U_{j,i};Y_{j,i}|U_{j-1,i}) -I(Y_j^{i-1};Y_{j,i}|W_1,\ldots,W_{j-1},Y_{j-1}^{i-1},Y_{j-3,i+1}^n)\nn\\
  &\quad+ I(Y_{j-3,i+1}^n;Y_{j,i}|W_1,\ldots,W_{j-1},Y_j^{i-1}) -I(Y_{j-2,i+1}^n;Y_{j,i}|W_1,\ldots,W_j,Y_j^{i-1})\Big)\nn\\
  &\overset{(b)}{=}n\epsilon_n+\sum_{i=1}^n \Big(I(U_{j,i};Y_{j,i}|U_{j-1,i}) -I(Y_j^{i-1};Y_{j,i}|W_1,\ldots,W_{j-1},Y_{j-1}^{i-1},Y_{j-3,i+1}^n)\nn\\
  &\quad+ I(Y_j^{i-1};Y_{j-3,i}|W_1,\ldots,W_{j-1},Y_{j-3,i+1}^n) -I(Y_{j-2,i+1}^n;Y_{j,i}|W_1,\ldots,W_j,Y_j^{i-1})\Big)\nn\\
  &\overset{(c)}{=}n\epsilon_n+\sum_{i=1}^n \Big(I(U_{j,i};Y_{j,i}|U_{j-1,i}) -I(Y_j^{i-1};Y_{j,i}|W_1,\ldots,W_{j-1},Y_{j-1}^{i-1},Y_{j-3,i+1}^n)\nn\\
  &\quad+ I(Y_j^{i-1};Y_{j-3,i}|W_1,\ldots,W_{j-1},Y_{j-3,i+1}^n,Y_{j-1}^{i-1})+ I(Y_{j-1}^{i-1};Y_{j-3,i}|W_1,\ldots,W_{j-1},Y_{j-3,i+1}^n)\nn\\
  &\quad-I(Y_{j-2,i+1}^n;Y_{j,i}|W_1,\ldots,W_j,Y_j^{i-1})\Big)\nn\\
  &\overset{(d)}{\leq} n\epsilon_n+\sum_{i=1}^n \Big(I(U_{j,i};Y_{j,i}|U_{j-1,i}) + I(Y_{j-1}^{i-1};Y_{j-3,i}|W_1,\ldots,W_{j-1},Y_{j-3,i+1}^n)\nn\\
  &\quad-I(Y_{j-2,i+1}^n;Y_{j,i}|W_1,\ldots,W_j,Y_j^{i-1})\Big)\nn\\
  &\overset{(e)}{=} n\epsilon_n+\sum_{i=1}^n \Big(I(U_{j,i};Y_{j,i}|U_{j-1,i}) + I(Y_{j-3,i+1}^n;Y_{j-1,i}|W_1,\ldots,W_{j-1},Y_{j-1}^{i-1})\nn\\
  &\quad-I(Y_{j-2,i+1}^n;Y_{j,i}|W_1,\ldots,W_j,Y_j^{i-1})\Big)
\end{flalign}
where $(a)$ is due to the independence of the messages and the Fano's inequality \eqref{eq:conv1}, $(b)$ is due to Csisz\'ar sum identity property, $(c)$ is due to the degradedness condition \eqref{eq:Markov} and the fact that
\begin{flalign}
  I&(Y_j^{i-1};Y_{j-3,i}|W_1,\ldots,W_{j-1},Y_{j-3,i+1}^n)\nn\\
  &=I(Y_j^{i-1};Y_{j-3,i}|W_1,\ldots,W_{j-1},Y_{j-3,i+1}^n,Y_{j-1}^{i-1})+ I(Y_{j-1}^{i-1};Y_{j-3,i}|W_1,\ldots,W_{j-1},Y_{j-3,i+1}^n),
\end{flalign}
the inequality $(d)$ is due to the degradedness condition \eqref{eq:Markov} and the fact that
\begin{flalign}
  -I&(Y_j^{i-1};Y_{j,i}|W_1,\ldots,W_{j-1},Y_{j-1}^{i-1},Y_{j-3,i+1}^n)+ I(Y_j^{i-1};Y_{j-3,i}|W_1,\ldots,W_{j-1},Y_{j-3,i+1}^n,Y_{j-1}^{i-1})\nn\\
  &=-I(Y_j^{i-1};Y_{j,i}|W_1,\ldots,W_{j-1},Y_{j-1}^{i-1},Y_{j-3,i+1}^n,Y_{j-3,i})\nn\\
  &\leq 0,
\end{flalign}
and $(e)$ is due to Csisz\'ar's sum identity property.

Following the intermediate step in \eqref{eq:555}, $H(W_j)$ is also upper bounded as follows:
\begin{flalign}\label{eq:666}
  H(W_j)&{\leq}n\epsilon_n+\sum_{i=1}^n \Big(I(U_{j,i};Y_{j,i}|U_{j-1,i}) -I(Y_j^{i-1};Y_{j,i}|W_1,\ldots,W_{j-1},Y_{j-1}^{i-1},Y_{j-3,i+1}^n)\nn\\
  &\quad+ I(Y_{j-3,i+1}^n;Y_{j,i}|W_1,\ldots,W_{j-1},Y_j^{i-1}) -I(Y_{j-2,i+1}^n;Y_{j,i}|W_1,\ldots,W_j,Y_j^{i-1})\Big).
\end{flalign}

Hence, substituting \eqref{eq:555} for $l\leq j\leq k$, and \eqref{eq:666} for $j=l-1$ into the first term in \eqref{eq:444}, we obtain,
\begin{flalign}\label{eq:777}
  \sum_{j=l-1}^k &H(W_{j})\nn\\
  &\leq n(k-l+2)\epsilon_n+\sum_{i=1}^n \sum_{j=l-1}^kI(U_{j,i};Y_{j,i}|U_{j-1,i})  \nn\\
  &\quad +I(Y_{l-4,i+1}^n;Y_{l-1,i}|W_1,\ldots,W_{l-2},Y_{l-1}^{i-1})-I(Y_{l-1}^{i-1};Y_{l-1,i}|W_1,\ldots,W_{l-2},Y_{l-2}^{i-1},Y_{l-4,i+1}^n)\nn\\
  &\quad -I(Y_{k-2,i+1}^n;Y_{k,i}|W_1,\ldots,W_k,Y_k^{i-1}).
%  &\leq n(k-l+2)\epsilon_n+\sum_{i=1}^n \Bigg(\bigg(\sum_{j=l-1}^kI(U_{j,i};Y_{j,i}|U_{j-1,i}) \bigg) \nn\\
%  &\quad-I(Y_{l-1}^{i-1};Y_{l-1,i}|W_1,\ldots,W_{l-2},Y_{l-2}^{i-1},Y_{l-4,i+1}^n) +I(Y_{l-4,i+1}^n;Y_{l-1,i}|W_1,\ldots,W_{l-2},Y_{l-1}^{i-1})\Bigg)\nn\\
\end{flalign}

We then bound the third term in \eqref{eq:444} for $3\leq l\leq k\leq K$ as follows:
\begin{flalign}\label{eq:888}
  -I&(W_{l-1}\ldots,W_k;Y_{l-2}^n|W_1,\ldots,W_{l-2})\nn\\
  &=\sum_{i=1}^n -I(W_{l-1}\ldots,W_k;Y_{l-2,i}|W_1,\ldots,W_{l-2},Y_{l-2,i+1}^n)\nn\\
  &=\sum_{i=1}^n -I(W_{l-1}\ldots,W_k,Y_k^{i-1};Y_{l-2,i}|W_1,\ldots,W_{l-2},Y_{l-2,i+1}^n) \nn\\
  &\quad+ I(Y_k^{i-1};Y_{l-2,i}|W_1,\ldots,W_k,Y_{l-2,i+1}^n)\nn\\
  &=\sum_{i=1}^n -I(W_{l-1}\ldots,W_k,Y_k^{i-1},Y_{l-2,i+1}^n;Y_{l-2,i}|W_1,\ldots,W_{l-2},Y_{l-4,i+1}^n)\nn\\
  &\quad +I(Y_{l-2,i+1}^n;Y_{l-2,i}|W_1,\ldots,W_{l-2},Y_{l-4,i+1}^n)+ I(Y_k^{i-1};Y_{l-2,i}|W_1,\ldots,W_k,Y_{l-2,i+1}^n)\nn\\
  &=\sum_{i=1}^n -I(W_{l-1}\ldots,W_k,Y_k^{i-1},Y_{k-2,i+1}^n;Y_{l-2,i}|W_1,\ldots,W_{l-2},Y_{l-2}^{i-1},Y_{l-4,i+1}^n)\nn\\
  &\quad +I(Y_{k-2,i+1}^n;Y_{l-2,i}|W_1,\ldots,W_k,Y_k^{i-1},Y_{l-2,i+1}^n)-I(Y_{l-2}^{i-1};Y_{l-2,i}|W_1,\ldots,W_{l-2},Y_{l-4,i+1}^n)\nn\\
  &\quad +I(Y_{l-2,i+1}^n;Y_{l-2,i}|W_1,\ldots,W_{l-2},Y_{l-4,i+1}^n)+ I(Y_k^{i-1};Y_{l-2,i}|W_1,\ldots,W_k,Y_{l-2,i+1}^n)\nn\\
  &=\sum_{i=1}^n -I(U_{k,i};Y_{l-2,i}|U_{l-2,i})\nn\\
  &\quad +I(Y_{k-2,i+1}^n;Y_{l-2,i}|W_1,\ldots,W_k,Y_k^{i-1},Y_{l-2,i+1}^n)-I(Y_{l-2}^{i-1};Y_{l-2,i}|W_1,\ldots,W_{l-2},Y_{l-4,i+1}^n)\nn\\
  &\quad +I(Y_{l-2,i+1}^n;Y_{l-2,i}|W_1,\ldots,W_{l-2},Y_{l-4,i+1}^n)+ I(Y_k^{i-1};Y_{l-2,i}|W_1,\ldots,W_k,Y_{l-2,i+1}^n)\nn\\
  &\overset{(a)}{=}\sum_{i=1}^n -I(U_{k,i};Y_{l-2,i}|U_{l-2,i})\nn\\
  &\quad -I(Y_{l-4,i+1}^n;Y_{l-2,i}|W_1,\ldots,W_{l-2},Y_{l-2}^{i-1})+I(Y_k^{i-1},Y_{k-2,i+1}^n;Y_{l-2,i}|W_1,\ldots,W_k,Y_{l-2,i+1}^n),
\end{flalign}
where $(a)$ is due to the following fact:
\begin{flalign}
  \sum_{i=1}^n& -I(Y_{l-2}^{i-1};Y_{l-2,i}|W_1,\ldots,W_{l-2},Y_{l-4,i+1}^n)+I(Y_{l-2,i+1}^n;Y_{l-2,i}|W_1,\ldots,W_{l-2},Y_{l-4,i+1}^n)\nn\\
  &=\sum_{i=1}^n-H(Y_{l-2,i}|W_1,\ldots,W_{l-2},Y_{l-4,i+1}^n)+H(Y_{l-2,i}|W_1,\ldots,W_{l-2},Y_{l-2}^{i-1},Y_{l-4,i+1}^n)\nn\\
  &\quad +H(Y_{l-2,i}|W_1,\ldots,W_{l-2},Y_{l-4,i+1}^n)-H(Y_{l-2,i}|W_1,\ldots,W_{l-2},Y_{l-2,i+1}^n)\nn\\
  &=\sum_{i=1}^n H(Y_{l-2,i}|W_1,\ldots,W_{l-2},Y_{l-2}^{i-1},Y_{l-4,i+1}^n)-H(Y_{l-2,i}|W_1,\ldots,W_{l-2},Y_{l-2}^{i-1})\nn\\
  &=\sum_{i=1}^n -I(Y_{l-4,i+1}^n;Y_{l-2,i}|W_1,\ldots,W_{l-2},Y_{l-2}^{i-1}).
\end{flalign}

Substituting \eqref{eq:777} and \eqref{eq:888} into \eqref{eq:444}, we obtain
\begin{flalign}\label{eq:999}
  n\sum_{j=l}^k& R_j\nn\\
  &\leq \sum_{j=l-1}^k H(W_{j})+n\epsilon_n-I(W_{l-1}\ldots,W_k;Y_{l-2}^n|W_1,\ldots,W_{l-2})\nn\\
  &\leq n(k-l+3)\epsilon_n+\sum_{i=1}^n \Bigg(\bigg(\sum_{j=l-1}^kI(U_{j,i};Y_{j,i}|U_{j-1,i}) \bigg)-I(U_{k,i};Y_{l-2,i}|U_{l-2,i}) \nn\\
  &\quad +I(Y_{l-4,i+1}^n;Y_{l-1,i}|W_1,\ldots,W_{l-2},Y_{l-1}^{i-1})-I(Y_{l-1}^{i-1};Y_{l-1,i}|W_1,\ldots,W_{l-2},Y_{l-2}^{i-1},Y_{l-4,i+1}^n)\nn\\
  %I(Y_k^{i-1};Y_{k,i}|W_1,\ldots,W_{k-1},Y_{k-1}^{i-1},Y_{k-3,i+1}^n)\nn\\
  &\quad -I(Y_{k-2,i+1}^n;Y_{k,i}|W_1,\ldots,W_k,Y_k^{i-1})\nn\\
  &\quad -I(Y_{l-4,i+1}^n;Y_{l-2,i}|W_1,\ldots,W_{l-2},Y_{l-2}^{i-1})+I(Y_k^{i-1},Y_{k-2,i+1}^n;Y_{l-2,i}|W_1,\ldots,W_k,Y_{l-2,i+1}^n)\Bigg),\nn\\
  &\overset{(a)}\leq n(k-l+3)\epsilon_n+\sum_{i=1}^n \Bigg(\bigg(\sum_{j=l-1}^kI(U_{j,i};Y_{j,i}|U_{j-1,i}) \bigg)-I(U_{k,i};Y_{l-2,i}|U_{l-2,i}) \Bigg),
\end{flalign}
where $(a)$ is due to the following two facts. The first fact is shown as follows:
\begin{flalign}
  \sum_{i=1}^n&I(Y_k^{i-1},Y_{k-2,i+1}^n;Y_{l-2,i}|W_1,\ldots,W_k,Y_{l-2,i+1}^n)-I(Y_{k-2,i+1}^n;Y_{k,i}|W_1,\ldots,W_k,Y_k^{i-1})\nn\\
  &=\sum_{i=1}^nH(Y_{l-2,i}|W_1,\ldots,W_k,Y_{l-2,i+1}^n)-H(Y_{l-2,i}|W_1,\ldots,W_k,Y_k^{i-1},Y_{k-2,i+1}^n)\nn\\
  &\quad -H(Y_{k,i}|W_1,\ldots,W_k,Y_k^{i-1})+H(Y_{k,i}|W_1,\ldots,W_k,Y_k^{i-1},Y_{k-2,i+1}^n)\nn\\
  &=H(Y_{l-2}^n|W_1,\ldots,W_k)-H(Y_{k}^n|W_1,\ldots,W_k)\nn\\
  &\quad +\sum_{i=1}^n H(Y_{k,i}|W_1,\ldots,W_k,Y_k^{i-1},Y_{k-2,i+1}^n,Y_{l-2,i})\nn\\
  &=-H(Y_{k}^n|W_1,\ldots,W_k,Y_{l-2}^n)+\sum_{i=1}^n H(Y_{k,i}|W_1,\ldots,W_k,Y_k^{i-1},Y_{k-2,i+1}^n,Y_{l-2,i})\nn\\
  &=\sum_{i=1}^n -H(Y_{k,i}|W_1,\ldots,W_k,Y_{l-2}^n,Y_k^{i-1})+H(Y_{k,i}|W_1,\ldots,W_k,Y_k^{i-1},Y_{k-2,i+1}^n,Y_{l-2,i})\nn\\
  &\leq 0.
\end{flalign}
The second fact is shown as follows:
\begin{flalign}
  \sum_{i=1}^n& I(Y_{l-4,i+1}^n;Y_{l-1,i}|W_1,\ldots,W_{l-2},Y_{l-1}^{i-1})-I(Y_{l-1}^{i-1};Y_{l-1,i}|W_1,\ldots,W_{l-2},Y_{l-2}^{i-1},Y_{l-4,i+1}^n)\nn\\
  &-I(Y_{l-4,i+1}^n;Y_{l-2,i}|W_1,\ldots,W_{l-2},Y_{l-2}^{i-1})\nn\\
  &=\sum_{i=1}^n H(Y_{l-1,i}|W_1,\ldots,W_{l-2},Y_{l-1}^{i-1})-H(Y_{l-1,i}|W_1,\ldots,W_{l-2},Y_{l-1}^{i-1},Y_{l-4,i+1}^n)\nn\\
  &-H(Y_{l-1,i}|W_1,\ldots,W_{l-2},Y_{l-2}^{i-1},Y_{l-4,i+1}^n)+ H(Y_{l-1,i}|W_1,\ldots,W_{l-2},Y_{l-4,i+1}^n,Y_{l-1}^{i-1})\nn\\
  &-H(Y_{l-2,i}|W_1,\ldots,W_{l-2},Y_{l-2}^{i-1})+H(Y_{l-2,i}|W_1,\ldots,W_{l-2},Y_{l-2}^{i-1},Y_{l-4,i+1}^n)\nn\\
  &=H(Y_{l-1}^n|W_1,\ldots,W_{l-2})-H(Y_{l-2}^n|W_1,\ldots,W_{l-2})\nn\\
  &\quad +\sum_{i=1}^n -H(Y_{l-1,i}|W_1,\ldots,W_{l-2},Y_{l-4,i+1}^n,Y_{l-2}^{i-1},Y_{l-2,i})\nn\\
  &=H(Y_{l-1}^n|W_1,\ldots,W_{l-2},Y_{l-2}^n)-\sum_{i=1}^n H(Y_{l-1,i}|W_1,\ldots,W_{l-2},Y_{l-4,i+1}^n,Y_{l-2}^{i-1},Y_{l-2,i})\nn\\
  &=\sum_{i=1}^n H(Y_{l-1,i}|W_1,\ldots,W_{l-2},Y_{l-2}^n,Y_{l-1}^{i-1})-H(Y_{l-1,i}|W_1,\ldots,W_{l-2},Y_{l-4,i+1}^n,Y_{l-2}^{i-1},Y_{l-2,i})\nn\\
  &\leq 0.
\end{flalign}

Furthermore, based on \eqref{eq:11}, we bound $\sum_{j=2}^K R_j$ as follows:
\begin{flalign}
  n\sum_{j=2}^KR_j&\leq n(k-1)\epsilon_n+ \sum_{j=2}^K\sum_{i=1}^n I(U_{j,i};Y_{j,i}|U_{j-1,i})\nn\\
  &\leq n(k-1)\epsilon_n +\sum_{j=2}^{K-1}\sum_{i=1}^n I(U_{j,i};Y_{j,i}|U_{j-1,i}) + \sum_{i=1}^n I(X_i;Y_{K,i}|U_{K-1,i}).
\end{flalign}

Based on \eqref{eq:999}, we bound $\sum_{j=l}^K R_j$ as follows:
\begin{flalign}
  n\sum_{j=l}^K R_j&\leq n(K-l+3)\epsilon_n+\sum_{i=1}^n \Bigg(\bigg(\sum_{j=l-1}^KI(U_{j,i};Y_{j,i}|U_{j-1,i}) \bigg)-I(U_{K,i};Y_{l-2,i}|U_{l-2,i}) \Bigg)\nn\\
  &=n(K-l+3)\epsilon_n+\sum_{i=1}^n \Bigg(\bigg(\sum_{j=l-1}^{K-1}I(U_{j,i};Y_{j,i}|U_{j-1,i}) \bigg)\nn\\
  &\quad+I(U_{K,i};Y_{K,i}|U_{K-1,i})-I(U_{K,i};Y_{l-2,i}|U_{l-2,i}) \Bigg)\nn\\
  &\overset{(a)}\leq n(K-l+3)\epsilon_n+\sum_{i=1}^n \Bigg(\bigg(\sum_{j=l-1}^{K-1}I(U_{j,i};Y_{j,i}|U_{j-1,i}) \bigg)\nn\\
  &\quad+I(X_{i};Y_{K,i}|U_{K-1,i})-I(X_{i};Y_{l-2,i}|U_{l-2,i}) \Bigg),
\end{flalign}
where $(a)$ is due to the Markov chain condition \eqref{eq:Markov2}.

The proof of the converse is then completed by defining a uniformly distributed random variable $Q\in\{1,\ldots,n\}$, and setting $U_k\triangleq (Q,U_{k,Q})$, $Y_k\triangleq Y_{k,Q}$, for $k\in [1:K]$, and $X\triangleq (Q,X_Q)$.

\renewcommand{\baselinestretch}{1}
\bibliographystyle{unsrt}
\bibliography{secrecy}

\end{document}